\documentclass[10.5pt,twocolumn]{article}

\usepackage[letterpaper,  margin=0.7in]{geometry}
\usepackage{url}
\usepackage{latexsym}
\usepackage{rotating}
\usepackage{graphicx}
\usepackage[breaklinks]{hyperref}  
\usepackage[hyphenbreaks]{breakurl}
\usepackage{bbm}
\usepackage[font=normalsize]{caption}

\usepackage[utf8]{inputenc}
 
\usepackage{amsmath}
\usepackage{amsfonts}
\usepackage{amssymb}
\usepackage{amsthm}
\usepackage{lipsum}
\usepackage{subcaption}
\usepackage{multicol}
\usepackage{multirow}
\usepackage{comment}

\usepackage[justification=centering]{caption}

\makeatletter
\newcommand*{\centerfloat}{
  \parindent \z@
  \leftskip \z@ \@plus 1fil \@minus \textwidth
  \rightskip\leftskip
  \parfillskip \z@skip}
\makeatother

\usepackage{xcolor}
\hypersetup{
    colorlinks,
    linkcolor={blue!40!black},
    citecolor={blue!40!black},
    urlcolor={blue!0!black}
}

\usepackage{multirow, booktabs}
\usepackage{siunitx}
\usepackage{enumitem}
\usepackage{mathtools}
\usepackage{titlesec}
\titleformat*{\section}{\large\bfseries}
\titleformat*{\subsection}{\normalsize\bfseries}
\titleformat*{\subsubsection}{\small\bfseries}
\newtheorem{prop}{Proposition}
\newtheorem*{prop*}{Proposition}
\numberwithin{prop}{section}

\numberwithin{equation}{section}

\theoremstyle{definition}
\newtheorem{definition}{Definition}[section]

\newtheorem{remark}{Remark}[section]

\DeclareMathOperator{\ind}{1{\hskip -2.5 pt}\textrm{I}}

\usepackage{chngcntr}
\counterwithin{figure}{section}
\counterwithin{table}{section}

\newtheorem{example}{Example}[section]

\title{\vspace{-0.8cm} \textbf{Minimum Rényi Entropy Portfolios\vspace{0.3cm}}}

\author{Nathan Lassance\thanks{Corresponding author: \texttt{nathan.lassance@uclouvain.be.}}~ and Frédéric Vrins\thanks{E-mail: \texttt{frederic.vrins@uclouvain.be}.} \\[0.3cm]
\small 
Université catholique de Louvain, Louvain Finance Center (LFIN) \\ \small and Center for Operations Research and Econometrics (CORE). \\ \small Voie du Roman Pays 34, 1348 Louvain-la-Neuve, Belgium. \\[0.5cm]
July 2018 \-- Available at \url{https://papers.ssrn.com/sol3/papers.cfm?abstract_id=2968660}
}
\allowdisplaybreaks
\date{}
\hypersetup{draft}

\titlelabel{\thetitle.\quad}

\begin{document}

\maketitle

\noindent \textbf{Abstract}

\smallskip 

\noindent Accounting for the non-normality of asset returns remains challenging in robust portfolio optimization. In this article, we tackle this problem by assessing the risk of the portfolio through the ``amount of randomness'' conveyed by its returns. We achieve this by using an objective function that relies on the exponential of \textit{Rényi entropy}, an information-theoretic criterion that precisely quantifies the uncertainty embedded in a distribution, accounting for higher-order moments. Compared to Shannon entropy, Rényi entropy features a parameter that can be tuned to play around the notion of uncertainty. A Gram-Charlier expansion shows that it controls the relative contributions of the central (variance) and tail (kurtosis) parts of the distribution in the measure. We further rely on a non-parametric estimator of the exponential Rényi entropy that extends a robust sample-spacings estimator initially designed for Shannon entropy. A portfolio selection application illustrates that minimizing Rényi entropy yields portfolios that outperform state-of-the-art minimum variance portfolios in terms of risk-return-turnover trade-off.

\vspace{0.2cm}


\noindent\rule[0.5ex]{\linewidth}{0.5pt}

\section{Introduction}\label{Intro}
In portfolio optimization, it is well-known that a high sensitivity of the optimal portfolio weights to estimation errors in the parameter inputs can render otherwise sound investment strategies largely sub-optimal out-of-sample (see \hyperref[Kolm]{Kolm et al. 2014} and references therein). This is in particular the case of the mean-variance portfolio of \hyperref[Markowitz]{Markowitz (1952)}: the optimal weights are very sensitive to estimation errors in the assets' expected return. To tackle this robustness issue, one can simply disregard the portfolio's expected return constraint, leading to the risk-based allocation framework (see e.g. \hyperref[Ardia]{Ardia et al. 2017}). \textit{Minimum risk portfolios} have in particular attracted investors' attention as ``the minimum-variance portfolio usually performs better out of sample than any other mean-variance portfolio \-- even when the Sharpe ratio or other performance measures related to \textit{both} the mean and variance are used for the comparison'' (\hyperref[DeMiguelNogales]{DeMiguel and Nogales 2009 p.560}).

Minimum risk portfolios are commonly built using the variance as risk measure. As the sample-based minimum variance portfolio is still quite vulnerable to estimation errors, various robust alternatives have been developed (see \hyperref[Fabozzi]{Fabozzi et al. 2010} and \hyperref[Scutella]{Scutellà and Recchia 2013} for reviews). Shrinkage estimation as introduced by \hyperref[Ledoit1]{Ledoit and Wolf (2003, 2004a,b)} has proved particularly appealing. Still, the problem remains that the variance is an adequate risk measure only for Gaussian distributions and is largely unaffected by increasing tail concentration (\hyperref[Vermorken]{Vermorken et al. 2012}). As a result, the minimum variance portfolio does not account for the non-normality of asset returns. Two main alternative approaches can be employed to deal with non-normality. First, one can minimize a downside risk measure, e.g. the minimum VaR and CVaR portfolios. However, such portfolios coincide with a mean-risk approach (\hyperref[Fabozzi]{Fabozzi et al. 2010}), producing robustness issues as for the mean-variance portfolio. Second, one can extend the Taylor expansion of the utility function to include the portfolio return's third and/or fourth moments (see e.g. \hyperref[Adcock]{Adcock 2014}). However, robustifying such higher-order portfolios is very challenging due to increased dimensionality and large sensitivity to outliers.

In this article, we propose a new (albeit natural) way of designing robust minimum risk portfolios that account for the non-normality of asset returns. We do so by minimizing the portfolio return's uncertainty measured via the exponential of \textit{Rényi entropy}, estimated within a robust $m$-spacings framework. The optimal portfolios so obtained are called \textit{minimum Rényi entropy portfolios}. Entropy is a well-known concept coming from information theory. It precisely aims at quantifying the uncertainty/amount of randomness conveyed by a distribution, embedding all higher-order moments (\hyperref[Cover]{Cover and Thomas 2006}). As a result, it is not surprising to notice that Shannon entropy (the most standard definition of entropy) has been recognized as an appealing measure in finance (\hyperref[Sbuelz]{Sbuelz and Trojani 2008}, \hyperref[Zhou]{Zhou et al. 2013}, \hyperref[Ormos]{Ormos and Zibriczky 2014}) portfolio management (\hyperref[Philippatos]{Philippatos and Wilson 1972}, \hyperref[Dionisio]{Dionisio et al. 2006}, \hyperref[Vermorken]{Vermorken et al. 2012}, \hyperref[Flores]{Flores et al. 2017}) and utility theory (\hyperref[Yang]{Yang and Qiu 2005}, \hyperref[Abbas]{Abbas 2006}, \hyperref[Jose]{Jose et al. 2008}). However, when using entropy to \textit{construct} optimal portfolios, the literature is so far limited to employing entropy as a penalty term besides a more standard cost function: one considers the weights as discrete probabilities, and uses their entropy as a penalty term to shrink them towards the equally-weighted solution (see \hyperref[Bera]{Bera and Park 2008}, \hyperref[Zhou]{Zhou et al. 2013}). Instead, we use the entropy of the portfolio return's \textit{distribution} (not that of the weights) as the risk-based cost function. Searching for the weights minimizing the latter amounts to minimize the returns' uncertainty and thus provides the minimum risk portfolio in the sense of information theory. 

In addition, we rely on Rényi entropy, an extension of Shannon entropy. It features a parameter $\alpha\in [0,\infty]$ that allows one to trade off the minimization of central and tail uncertainty. We argue in particular in favor of setting $\alpha \in [0,1]$ as, then, Rényi entropy has natural connections with measures of spread (as shown by the extended Chebyshev's inequality of \hyperref[Campbell]{Campbell 1966}) and with a minimum variance-kurtosis objective (as shown by a novel Gram-Charlier expansion of the measure). The empirical results support that choice as well.

Our contribution is organized as follows. Section \ref{Section2Renyi} explores the theoretical properties of the exponential Rényi entropy and makes the link with the notion of risk. Section \ref{Section3MRE} follows with the minimum Rényi entropy portfolio and its connections with higher-order moments. Section \ref{Section4Est} derives a robust $m$-spacings estimator of the measure and studies its consistency and robustness. We design and perform an empirical out-of-sample performance study of the proposed method in Section \ref{Section5Empirical}. Minimum Rényi entropy portfolios are shown to outperform standard minimum risk portfolios in terms of risk-adjusted performance, while achieving a reasonable level of turnover for values of $\alpha$ close to zero. Section \ref{Conclusion} concludes. 

For conciseness, the proofs of the different propositions are reported in the Online Resources.

\section{Exponential Rényi entropy and risk measurement}\label{Section2Renyi}
We start this section by introducing the \textit{Rényi entropy}, a flexible measure that quantifies the uncertainty of a random variable from its distribution. It encompasses the well-known \textit{Shannon entropy} which is recovered as a special case. We then show how its exponential transform can be thought of as a \textit{deviation risk measure}. A discussion of the impact of Rényi's $\alpha$ parameter closes the section, arguing in particular that setting $\alpha \in [0,1]$ should be favored in our portfolio selection context. 

In the sequel, we respectively denote $F_X$ and $f_X$ the cdf and pdf of a random variable $X$. We are exclusively interested in continuous distributions.

\subsection{Shannon and Rényi entropy}\label{sectionRenyi}
The entropy of a random variable $X$ commonly refers to its Shannon entropy, first introduced by \hyperref[Shannon]{Shannon (1948)}, giving birth to a new scientific discipline: \textit{information theory}. It is defined as  
\begin{equation}\label{diffshannon}
H(X) := H[f_X] := -\mathbb{E}\big(\ln f_X(X)\big)\;.
\end{equation}
This measure is known to quantify the amount of randomness embedded in $X$. For instance, when $X$ is a continuous random variable with bounded support, this quantity is maximized for the uniform distribution, which is the most uncertain one. Shannon entropy embeds many important properties. We refer to \hyperref[Cover]{Cover and Thomas (2006)} for an extended treatment.

\hyperref[Renyi]{Rényi (1961)} proposed a generalization of Shannon entropy in \eqref{diffshannon} with the help of a parameter $\alpha\in\mathbb{R}^+$. The idea was to consider a generalized averaging of $-\ln f_X$, leading to the following definition:
\begin{equation}\label{renyi}
H_\alpha (X) := H_\alpha [f_X] := \frac{1}{1-\alpha}\hspace{0.02cm}\ln \mathbb{E}\big(f_X^{\alpha-1}(X)\big)\;,
\end{equation}
whenever the expectation exists. Shannon entropy is recovered as a special case in the sense that
\begin{equation}
\lim_{\alpha\to 1}H_\alpha(X):=H_1(X)=H(X)\;.
\end{equation}
Just like Shannon entropy, Rényi entropy enjoys interesting properties. However, its exponential transform has more natural properties in the context of risk. The next section is dedicated to a more detailed analysis of the exponential Rényi entropy and its connection with deviation risk measures. 

\subsection{Exponential Rényi entropy}\label{deviationriskmeasure}
We denote by $H_\alpha^{\exp}$ the exponential Rényi entropy, which, from \eqref{renyi}, reads as
\begin{equation}\label{exprenyiformula}
H_\alpha^{\exp}(X) := \exp(H_\alpha(X)) = \left(\int (f_X(x))^\alpha dx\right)^{\frac{1}{1-\alpha}}\;.
\end{equation}
This quantity was first introduced by \hyperref[Campbell]{Campbell (1966)} who studied its relevance as a measure of spread of a distribution for $\alpha \in [0,1]$. We come back to the link between $H_\alpha^{\exp}$ and measures of spread in Section \ref{AlphaZeroOne}. In this article, we apply this measure to the construction of minimum risk portfolios (see Section \ref{Section3MRE}).


\subsubsection{Properties}

From the properties of Rényi entropy (see \hyperref[Koski]{Koski and Persson 1992}, \hyperref[Johnson]{Johnson and Vignat 2007}, \hyperref[Pham1]{Pham et al. 2008}), $H_\alpha^{\exp}$ obeys the below properties. 
\begin{prop}\label{RenyiProperties} Let $c$ be a real constant. $H_\alpha^{\exp}(X)$ satisfies the following properties:
\begin{itemize}
\item[(i)] translation-invariance: 
\begin{equation*}
H_\alpha^{\exp}(X+c) = H_\alpha^{\exp}(X)\;;
\end{equation*}
\item[(ii)] scaling property: 
\begin{equation*}
H_\alpha^{\exp}(c X) = |c| H_\alpha^{\exp}(X)\;;
\end{equation*}
\item[(iii)] it is non-increasing and continuous in $\alpha$\;.
\end{itemize}
\end{prop}

\subsubsection{Connection with deviation risk measures}

Quantifying uncertainty, it is appealing to use the exponential Rényi entropy as a deviation risk measure, introduced by \hyperref[Rockafellar2]{Rockafellar et al. (2006)}.

\begin{definition}
A \textit{deviation risk measure} is any functional $\mathcal{D}: L^p(\Omega) \rightarrow [0,\infty]$ satisfying:\footnote{$L^p(\Omega)$ is the space of random variables defined on the support set $\Omega$ having finite $p^{\text{th}}$ moment.} 
\begin{enumerate}
\item[\textit{(i)}] \textit{Positivity}: $\mathcal{D}(X) > 0$ for all non-constant $X$, and $\mathcal{D}(X) = 0$ for any constant $X$\;;
\item[\textit{(ii)}] \textit{Positive homogeneity}: $\mathcal{D}(c X) = c\mathcal{D}(X) \ \forall c > 0$\;;
\item[\textit{(iii)}] \textit{Translation-invariance}: $\mathcal{D}(X+c) = \mathcal{D}(X) \ \forall c \in \mathbb{R}$\;;
\item[\textit{(iv)}] \textit{Sub-additivity}: $\mathcal{D}(X+Y) \leqslant \mathcal{D}(X) + \mathcal{D}(Y)$\;.
\end{enumerate}
\end{definition}

Let us show that $H_\alpha^{\exp}$ fulfills the first three properties of deviation risk measures (sub-additivity is dealt with in next section).

Positive homogeneity $(ii)$ and translation invariance $(iii)$ result from the properties of $H_\alpha^{\exp}$ in Proposition \ref{RenyiProperties}. Regarding positivity $(i)$, $H_\alpha^{\exp}(X)$ is strictly positive if $X$ is non-constant from the positivity of the density $f_X$. To see that it is null if $X$ is constant, let us compute $H_\alpha^{\exp}(k)$ where $k$ is a constant by computing the limit of $H_\alpha^{\exp}(k+cX)$ as $c$ tends to zero for a given random variable $X$ of finite entropy:
\begin{equation*}
H_\alpha^{\exp}(k)=\lim_{c\downarrow 0}H_\alpha^{\exp}(k+cX)=\lim_{c\downarrow 0} cH_\alpha^{\exp}(X) = 0\;.
\end{equation*}

\subsubsection{The sub-additivity property}\label{subadditivity}

In this section, we begin with underlining that whereas $H_\alpha^{\exp}$ is expected to be sub-additive for most cases encountered in portfolio selection, it is not, \textit{strictly speaking}, sub-additive. 

\begin{prop}\label{PropSubAdd}
$H_\alpha^{\exp}$ is, generally speaking, not a sub-additive measure.
\end{prop}
\begin{proof}
Appendix \ref{AppendixSubAdd} gives three analytical counter-examples to sub-additivity using pairs of random variables $(X,Y)$ : a pair of independent one-sided (Lévy), a pair of independent two-sided bimodal (Gaussian mixtures) based on the entropy bounds derived in \hyperref[Vrins]{Vrins et al. (2007)}, and eventually a pair of comonotonic random variables. 
\end{proof}

We stress that this proposition contradicts some statements recently made in the portfolio management literature (see e.g. \hyperref[Flores]{Flores et al. 2017}).\footnote{We are grateful to the authors of the aforementioned paper for discussion about the provided counter-examples.}

It is however worth stressing that those counter-examples are atypical in portfolio applications. The Lévy distribution for instance is extremely heavy-tailed: none of the moments are defined, and asset returns exibit much lighter tails in practice (\hyperref[Cont]{Cont 2001}). Similarly, multi-modal distributions and comonotonicity are behaviours that rarely arise in portfolio management. 


In fact, just like for the Value-at-Risk (\hyperref[Danielsson]{Danielsson et al. 2013}), sub-additivity of the exponential Rényi entropy can be reasonably assumed to hold in the specific context of portfolio optimization. For instance, the exponential Rényi entropy of $Z\sim\mathcal{N}(\mu,\sigma)$ collapses to $H_\alpha^{\exp}(Z) = \sigma\sqrt{2\pi\alpha^{1/(\alpha-1)}}$ (\hyperref[Koski]{Koski and Persson 1992}). From the stability of the Gausian distribution, the sub-additivity property is equivalent to $\sigma_{X+Y} \leqslant \sigma_X + \sigma_Y$, which in turns holds from the sub-additivity of the standard deviation (\hyperref[Artzner]{Artzner et al. 1999}). 

%
%


However, this particular case is quite restrictive as asset returns are typically not well described by the Gaussian distribution. A more appealing candidate to model the fat tails observed in asset returns is the general class of \textit{elliptical} (a.k.a. \textit{radial}) distributions, that has many applications in mathematical finance and portfolio theory, see e.g. \hyperref[Chen]{Chen et al. (2011)}. Elliptical distributions comprise, among others, the Gaussian, Student's $t$, Cauchy and Laplace distributions. As the next proposition shows, the exponential Rényi entropy is sub-additive when $(X,Y)$ is distributed according to an elliptical distribution, providing a broader and more realistic sufficient condition than the Gaussian setting.

\begin{prop}\label{SubAddElliptical}
Let $(X,Y) \sim \textup{El}(\boldsymbol{\mu},\Sigma,g_2)$ with $\textup{El}(\boldsymbol{\mu},\Sigma,g_2)$ a bivariate elliptical distribution, i.e.
\begin{equation*}\label{BivariateT}
f_{X,Y}(\mathbf{x}) = |\Sigma|^{-1/2}g_2\big((\mathbf{x}-\boldsymbol{\mu})'\Sigma^{-1}(\mathbf{x}-\boldsymbol{\mu})\big)\;,
\end{equation*}
where $g_2$ is a non-negative density generator function and $|\Sigma|$ is the absolute value of the determinant of $\Sigma=\begin{psmallmatrix} 
\sigma_X^2 & \rho\sigma_X\sigma_Y \\
\rho\sigma_X\sigma_Y & \sigma_Y^2 
\end{psmallmatrix}$, the scale matrix of $(X,Y)$. Then, $H_\alpha^{\exp}$ is sub-additive for the pair $(X,Y)$.
\end{prop}
\begin{proof}
See Appendix \ref{ProofElliptical}.
\end{proof}

\begin{remark}
Proposition \ref{SubAddElliptical} can be extended to any dimension, meaning that, if $\boldsymbol{X}=(X_1, \dots, X_n)$ $\sim \text{El}(\boldsymbol{\mu},\Sigma,g_n)$, then  $H_\alpha^{\text{exp}}\big(\sum_{i=1}^n X_i\big) \leqslant \sum_{i=1}^n H_\alpha^{\text{exp}}(X_i)$. Combined with the positive homogeneity property, this means that $H_\alpha^{\exp}$ is sub-additive at the portfolio level, i.e. denoting $(w_1, \dots, w_n)$ positive portfolio weights, we have $H_\alpha^{\exp}\big(\sum_{i=1}^n w_iX_i \big) \leqslant \sum_{i=1}^n w_i H_\alpha^{\exp}(X_i)$.
\end{remark}

\subsection{Exponential Rényi entropy as a flexible risk measure}\label{choicealpha}

This section explains how $\alpha$ allows to tune the relative contributions of the central and tail parts of the distribution, leading to different definitions of risk.

To show this, consider the two extreme cases $\alpha=0$ and $\alpha=\infty$. As the next proposition shows, when $\alpha=0$, $H_\alpha^{\exp}(X)$ measures the spread of $X$, while when $\alpha=\infty$, $H_\alpha^{\exp}(X)$ is given by the inverse of the supremum of $f_X$ (see \hyperref[Koski]{Koski and Persson 1992}).

\begin{prop}
Let $X$ be a continuous random variable, then $H_0^{\exp}(X) := \lim\limits_{\alpha \downarrow 0}H_\alpha^{\exp}(X)$ and $H_\infty^{\exp}(X) := \lim\limits_{\alpha \rightarrow \infty}H_\alpha^{\exp}(X)$ read as
\begin{gather}
H_0^{\exp}(X)  = \mathcal{L}(\Omega)\;, \\
H_\infty^{\exp}(X) = 1/\sup f_X\;,
\end{gather}
where $\mathcal{L}(\Omega)$ is the Lebesgue measure of the support set of $X$, $\Omega:=\{ x : f_X(x)>0 \}$.
\end{prop}



As one can see, changing $\alpha$ modifies the way we measure entropy, i.e. uncertainty, and so the risk. Taking $\alpha=0$ amounts to measure risk by the support of the distribution, while taking $\alpha=\infty$ amounts to measure risk by the maximal probability. By minimizing the portfolio return's entropy, as we propose in the next section, one can therefore minimize the density range on the $x$-axis with $\alpha=0$ or maximize the density range on the $y$-axis with $\alpha=\infty$. $H_0^{\exp}$ focuses only on extreme values (low entropy = low distance between extreme values), while $H_\infty^{\exp}$ focuses only on the most likely outcomes (low entropy = high maximal probability), and so, in the symmetric unimodal case, on the center of the distribution. 

From those two extreme cases, it results that, in portfolio selection applications, taking $\alpha$ too large is not desirable because $H_\alpha^{\exp}$ will barely be affected by tail events, which is the criticism that is made about the variance. Conversely, by decreasing $\alpha$, we assign more similar ``weight'' to all events, hence increasing the relative importance of tail events compared to events around the mode. 


\begin{example}

To illustrate this effect of $\alpha$, Figure \ref{StudentTailWeight} shows how $H_\alpha^{\exp}(X)$, $X \sim$ $t$-Student($\nu$), evolves with $\nu$ for different values of $\alpha$. From \hyperref[Zografos1]{Zografos and Nadarajah (2005)}, $H_\alpha^{\exp}(X)$ expresses as
\begin{equation}\label{RenyiStudent}
H_\alpha^{\exp}(X) = (\pi\nu)^{\frac{1-\alpha}{2}}\Bigg(\frac{\Gamma\big(\frac{\nu+1}{2}\big)}{\Gamma\big(\frac{\nu}{2}\big)}\Bigg)^\alpha \frac{\Gamma\big(\frac{\alpha(\nu+1)}{2}-\frac{1}{2}\big)}{\Gamma\big(\frac{\alpha(\nu+1)}{2}\big)}\;.
\end{equation}
As one can see, when going from $\alpha=2$ to $\alpha=0.4$, the sensitivity to the increase of tail uncertainty is indeed increasingly visible.

\begin{figure*}
\centering
\caption{The sensitivity of $H_\alpha^{\exp}$ to the tail uncertainty of $X \sim t$-Student($\nu$) increases when $\alpha$ decreases. \textit{Notes:} We report $H_\alpha^{\exp}(X)$ in \eqref{RenyiStudent} for different values $\alpha$ as a function of the number of degrees of freedom $\nu$ of $X$.}
\vspace{-0.25cm}
\hspace{1cm}\includegraphics[width=0.8\textwidth]{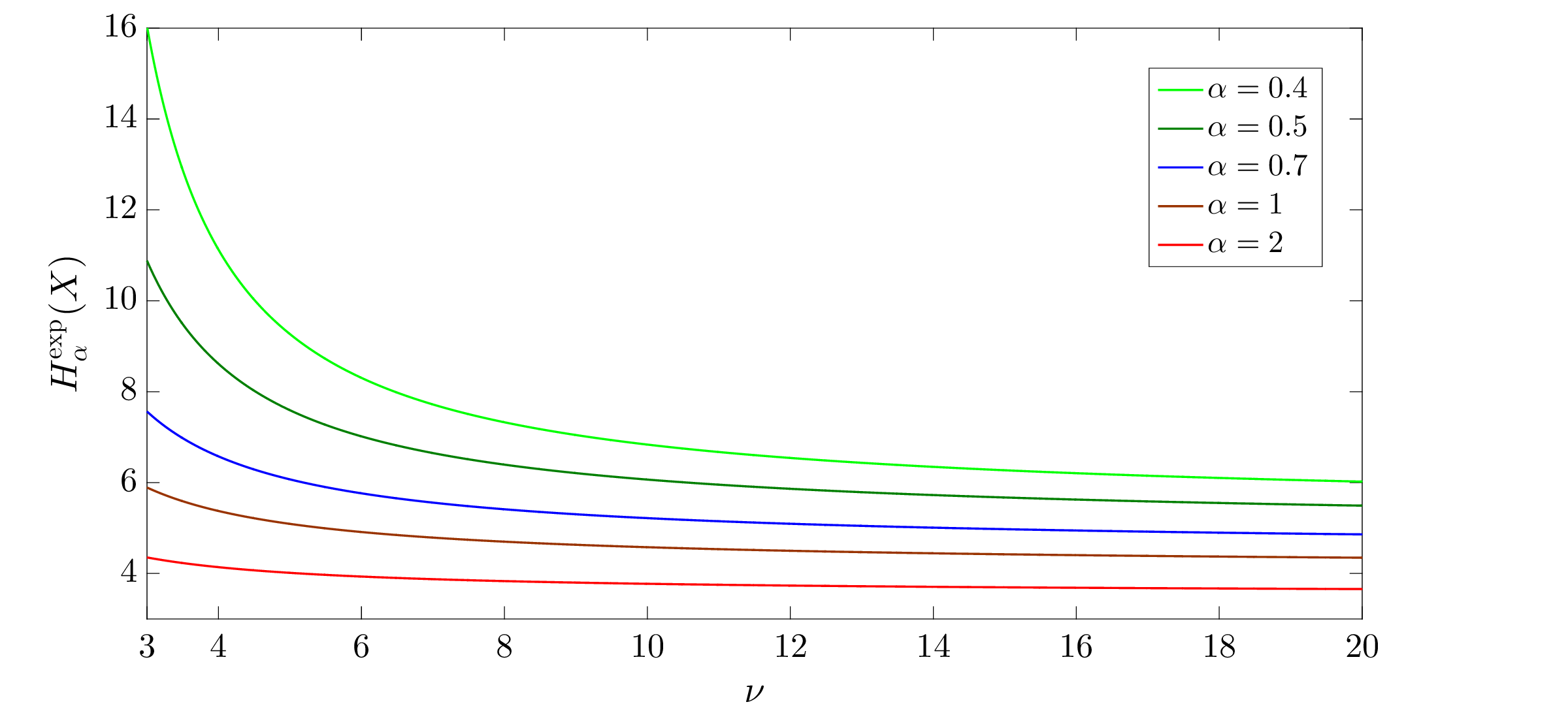}
\label{StudentTailWeight}
\vspace{-0.2cm}
\end{figure*}

\end{example}

\subsection{Appeal of the $\alpha\in [0,1]$ case in a portfolio selection context}\label{AlphaZeroOne}

The previous section argued how decreasing the value of $\alpha$ allows one to obtain a measure of entropy that is increasingly more affected by tail events. In this section, we show more specifically that investors should favor setting $\alpha \in [0,1]$, in which case $H_\alpha^{\exp}$ provides an appealing extension of the variance as a risk criterion.

First, for $\alpha \in [0,1]$, $H_\alpha^{\exp}$ has close connections with measures of spread. By minimizing the variance, investors ensure that most of the probability distribution of the portfolio return is concentrated in some small interval around the mean. This is established by Chebyshev's inequality which, given the set $A_k = \{x \in \Omega \ | \ |x-\mathbb{E}(X)|\geq k \}$, says that
\begin{equation*}
\mathbb{P}(X\in A_k) \leq \frac{\mathbb{V}\text{ar}(X)}{k^2}\;.
\end{equation*}
Similarly, for $\alpha \in [0,1]$, a small value of $H_\alpha^{\exp}(X)$ entails that most of the probability distribution of $X$ is concentrated on a set of small Lebesgue measure. This is determined by \hyperref[Campbell]{Campbell (1966)}'s extended Chebyshev's inequality. 

\begin{prop}\label{PropChebyshev}
Let $X$ be a continuous random variable whose Rényi entropy is defined for $\alpha \in [0,1]$. Then, given $\alpha \in [0,1]$ and $A_k' = \{ x\in \Omega \ | \ f_X(x) \leq k \}$, the following inequality holds:
\begin{equation}\label{IneqRenyi}
\mathbb{P}(X\in A_k') \leq \big(kH_\alpha^{\exp}(X)\big)^{1-\alpha}\;.
\end{equation}
\end{prop}
\noindent This inequality is more general than Chebyshev's inequality as it does not only deal with the absolute deviation around the mean, but instead relates the spread in terms of the size of the set on which most of the probability density is situated. For a unimodal random variable $X$ with $\Omega = \mathbb{R}$ and $f_X(x) \rightarrow 0$ as $x \rightarrow \pm \infty$, which is common for asset returns, then there are only two values of $x$ for which $f_X(x) = k$ (if $k < \text{mode}(X)$). Denoting them $x_k^-$ and $x_k^+$, we have $x_k^- < \text{mode}(X) < x_k^+$ and \eqref{IneqRenyi} means that
\begin{equation*}
\mathbb{P}(x_k^- < X < x_k^+) \geq 1-\big(kH_\alpha^{\exp}(X)\big)^{1-\alpha} \rightarrow 1
\end{equation*}
as $H_\alpha^{\exp}(X) \rightarrow 0$. In other words, if $H_\alpha^{\exp}(X)$ is small, the probability that $X$ is concentrated on a small interval around its mode is close to one. 


A second argument in favor of setting $\alpha \in [0,1]$ is related to the Gram-Charlier expansion of Rényi entropy derived in Section \ref{GramCharlier}. The expansion will show that, when $\alpha \in [0,1]$, the coefficient in front of the kurtosis of $X$ is positive (and instead negative for $\alpha > 1$) and so that an increase in kurtosis decreases the Rényi entropy, as desired by investors.

Third, the empirical results presented in Section \ref{EmpiricalResults} display a largely better performance of the minimum Rényi entropy portfolio when $\alpha \in [0,1]$ as well.

\section{Minimum Rényi entropy portfolio}\label{Section3MRE}
Given the good match between the theoretical properties of $H_\alpha^{\exp}$ and the desirable features of portfolio selection criteria, we use this measure as an objective function to design investment strategies. In particular, we propose to construct a minimum risk portfolio, called the \textit{minimum Rényi entropy (MRE) portfolio}, that minimizes the exponential Rényi entropy of the portfolio return. We denote by $P$ the portfolio return such that $P =\boldsymbol{w}'\boldsymbol{X}= \sum_{i=1}^n w_i X_i$ where $\boldsymbol{w}=(w_1,\dots,w_n)'$ is the vector of portfolio weights and $\boldsymbol{X}=(X_1,\dots,X_n)'$ is the random vector of assets' return. 

\subsection{Definition}\label{DefinitionROpt}
The MRE portfolio over an investment set of $n$ assets for a given $\alpha$ is defined as
\begin{equation}\label{ROptProgram}
\boldsymbol{w}_{\alpha}^\star := \arg\underset{\boldsymbol{w}\in \mathcal{W}}{\min}\hspace{0.05cm}H_\alpha^{\exp}(P)\;,
\end{equation}
where $\mathcal{W}$ is a set of constraints on $\boldsymbol{w}$, including the full investment constraint $\boldsymbol{1}_n'\boldsymbol{w}=1$. 

Note that, being affected by higher-order moments that are non-convex functions of the weights (\hyperref[Jurczenko]{Jurczenko and Maillet 2006}), the optimization program in \eqref{ROptProgram} may not  necessarily be convex, i.e. feature only one local optimum. Hence, in solving for the MRE portfolio, one must ideally resort to global optimization techniques rather than standard local optimizers. We come back to this matter in Section \ref{Optimization}.

\subsection{Connection with moments-based portfolios: A Gram-Charlier expansion}\label{GramCharlier}

Traditional minimum risk portfolios are built from specific moments of the portfolio return, typically the variance, leading to the minimum variance portfolio, and possibly higher-order moments as in e.g. \hyperref[Martellini]{Martellini and Ziemann (2010)}, \hyperref[Adcock]{Adcock (2014)} and \hyperref[Vanduffel]{Vanduffel and Yao (2017)}, that we call \textit{higher-order portfolios}. 

In the classical Markowitz Gaussian setting, the MRE portfolio coincides with the minimum variance portfolio as there is a one-to-one correspondence between $H_\alpha^{\exp}$ and the variance for Gaussian random variables. 

In a more general setting however, the MRE portfolio is more attractive than the minimum variance one as it accounts for the uncertainty coming from higher-order moments. To see this, it is useful to derive a truncated Gram-Charlier (GC) expansion of Rényi entropy. 

\begin{prop}\label{PropGC}
Let $X\in L^4(\Omega)$ and note $\tilde{X} = (X-\mathbb{E}(X))/\sqrt{\mathbb{V}\textup{ar}(X)}$ its standardized copy. Define $\mathbb{S}\textup{kew}(X)=\mathbb{E}(\tilde{X}^3)$ and $\mathbb{K}\textup{urt}(X)=\mathbb{E}(\tilde{X}^4)-3$. Then, the truncated GC expansion of $H_\alpha(X)$, denoted $H_\alpha^{GC}(X)$, writes as
\begin{equation}\label{GCApprox}
\begin{aligned}
H_\alpha^{GC}(X) := \hspace{0.05cm} &H_\alpha\big[\mathcal{N}\big(0,\sqrt{\mathbb{V}\textup{ar}(X)}\big)\big] + k_1(\alpha) \mathbb{K}\textup{urt}(X) \\  &+ k_2(\alpha) \mathbb{S}\textup{kew}(X)^2 + k_3(\alpha) \mathbb{K}\textup{urt}(X)^2\;,
\end{aligned}
\end{equation}
with coefficients
\begin{equation}\label{CoeffGC}
\begin{aligned}
k_1(\alpha) &= \frac{1-\alpha}{8\alpha}\;, \\
k_2(\alpha) &= -\frac{3\alpha^2-6\alpha+5}{24\alpha^{3/2}}\;, \\
k_3(\alpha) &= - \frac{3\alpha^4-12\alpha^3+42\alpha^2-60\alpha+35}{384\alpha^{5/2}}\;.
\end{aligned}
\end{equation}
\end{prop}
\begin{proof}
See Appendix \ref{AppendixGC}.
\end{proof}

The three coefficients $k_1(\alpha)$, $k_2(\alpha)$ and $k_3(\alpha)$ are displayed in Figure \ref{GramCharlierRenyi}.

\begin{figure*}[ht]
\centering
\caption{Coefficients $k_1(\alpha)$, $k_2(\alpha)$ and $k_3(\alpha)$ of the Gram-Charlier expansion of Rényi entropy. \textit{Notes:} The Gram-Charlier expansion and expressions for the coefficients are displayed in Equations \eqref{GCApprox} and \eqref{CoeffGC}.}
\vspace{-0.3cm}
\hspace{1.1cm}\includegraphics[width=0.8\textwidth]{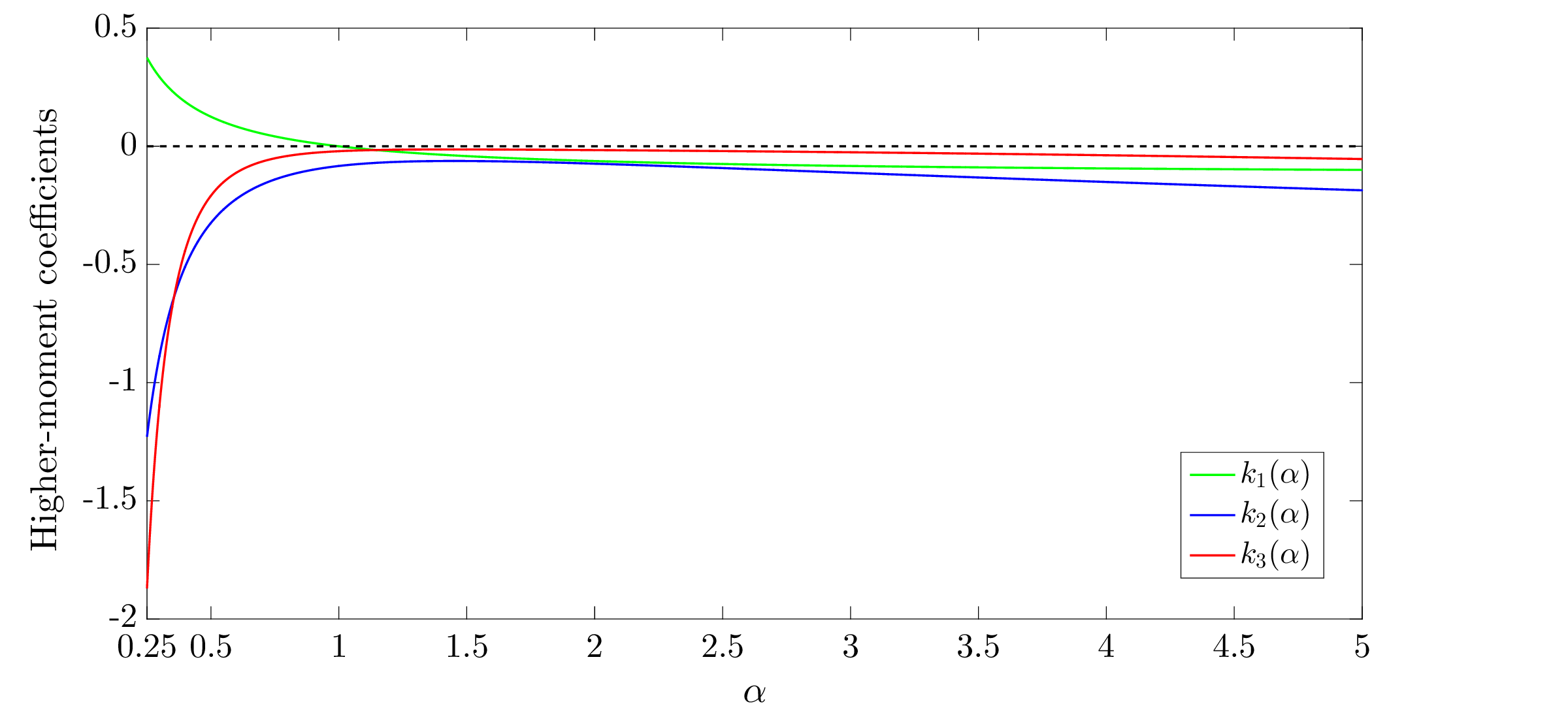}
\label{GramCharlierRenyi}
\vspace{-0.2cm}
\end{figure*}

Setting $\alpha=1$, we get the GC expansion $H_1^{GC}(X) = H_1\big[\mathcal{N}\big(0,\sqrt{\mathbb{V}\textup{ar}(X)}\big)\big] - \frac{1}{12}\mathbb{S}\textup{kew}(X)^2 \\ - \frac{1}{48}\mathbb{K}\textup{urt}(X)^2$ derived in \hyperref[Hyvarinen]{Hyv\"arinen et al. (2001)}. Hence, the ability to control for kurtosis is a notable advantage of Rényi entropy over Shannon entropy, as in the latter case $k_1(1)=0$. 

The connection between the MRE and higher-order portfolios is now made explicit. We have $H_\alpha\big[\mathcal{N}\big(0,\sqrt{\mathbb{V}\textup{ar}(P)}\big)\big]=H_\alpha\big[\mathcal{N}(0,1)\big]+\frac{1}{2}\ln \mathbb{V}\textup{ar}(P)$, which yields\footnote{Minimizing $H_\alpha^{\exp}(P)$ or $H_\alpha(P)$ is equivalent as $\exp(x)$ is a monotonically increasing function.}
\begin{equation}
\begin{aligned}
\boldsymbol{w}^\star_\alpha \approx \arg \min_{\boldsymbol{w}\in \mathcal{W}} &\frac{1}{2}\ln\mathbb{V}\textup{ar}(P) + k_1(\alpha) \mathbb{K}\textup{urt}(P) \\ &+ k_2(\alpha) \mathbb{S}\textup{kew}(P)^2 + k_3(\alpha) \mathbb{K}\textup{urt}(P)^2\;.
\end{aligned}
\end{equation}
When $f_P$ is close to a Gaussian, the main contributing higher-order term will be $k_1(\alpha) \mathbb{K}\textup{urt}(P)$. When $\alpha < 1$, $k_1(\alpha)>0$ and so the MRE portfolio is similar to a minimum variance-kurtosis portfolio, which, as noted by \hyperref[Martellini]{Martellini and Ziemann (2010)}, is a well-performing higher-order portfolio as estimators for even moments are less noisy than estimators for odd moments. When $\alpha>1$ however, $k_1(\alpha)<0$ and so the effect is reversed. In line with investors' preferences for kurtosis, setting $\alpha \in [0,1]$ is thus more natural, as we noted in Section \ref{AlphaZeroOne}. 

Hence, in line with Section \ref{choicealpha}, by playing with $\alpha$ one trades off the minimization of the central (variance) and tail (kurtosis) uncertainty, i.e. of the first two even moments.

\begin{example}
Consider $n=2$ assets $X_1=X$ $\perp$ $X_2=Y$ that follow a zero-mean Student's $t$ distribution with $(\sigma_X,\nu_X)=(0.3,10)$ and $(\sigma_Y,\nu_Y)=(0.2,6)$. We build a portfolio $P=wX + (1-w)Y$ and evaluate $H_\alpha^{\exp}(P)$ by numerical integration. On Figure \ref{EffectOfKurtosis}, we display how $H_\alpha^{\exp}(P)$, $\sqrt{\mathbb{V}\text{ar}(P)}$ and $\mathbb{K}\text{urt}(P)$ depend on $w$. As we can see, when $\alpha$ is high enough, $w^\star$ is close to the minimum variance solution because $\sigma_X > \sigma_Y$ and that mostly central events matter when $\alpha$ is high. However, the more $\alpha$ decreases, the more important is the impact of the fatter tails of $Y$ ($\nu_Y < \nu_X$) and so the more $w^\star$ approaches the minimum kurtosis solution. 
\end{example}

\begin{figure*}
\centering
\caption{The minimum Rényi entropy portfolio trades off the minimization of variance and kurtosis. \textit{Notes:}  $X \perp Y$ follow a zero-mean Student's $t$ distribution with $(\sigma_X,\nu_X)=(0.3,10)$ and $(\sigma_Y,\nu_Y)=(0.2,6)$. We consider a portfolio $P=wX + (1-w)Y$ and plot its standard deviation $\sqrt{\mathbb{V}\text{ar}(P)}$, excess kurtosis $\mathbb{K}\text{urt}(P)$ and exponential Rényi entropy $H_\alpha^{\exp}(P)$ for different values of $\alpha$ and $w\in[0,1]$.}
\vspace{-0.3cm}
\hspace{0.5cm}\includegraphics[width=0.8\textwidth]{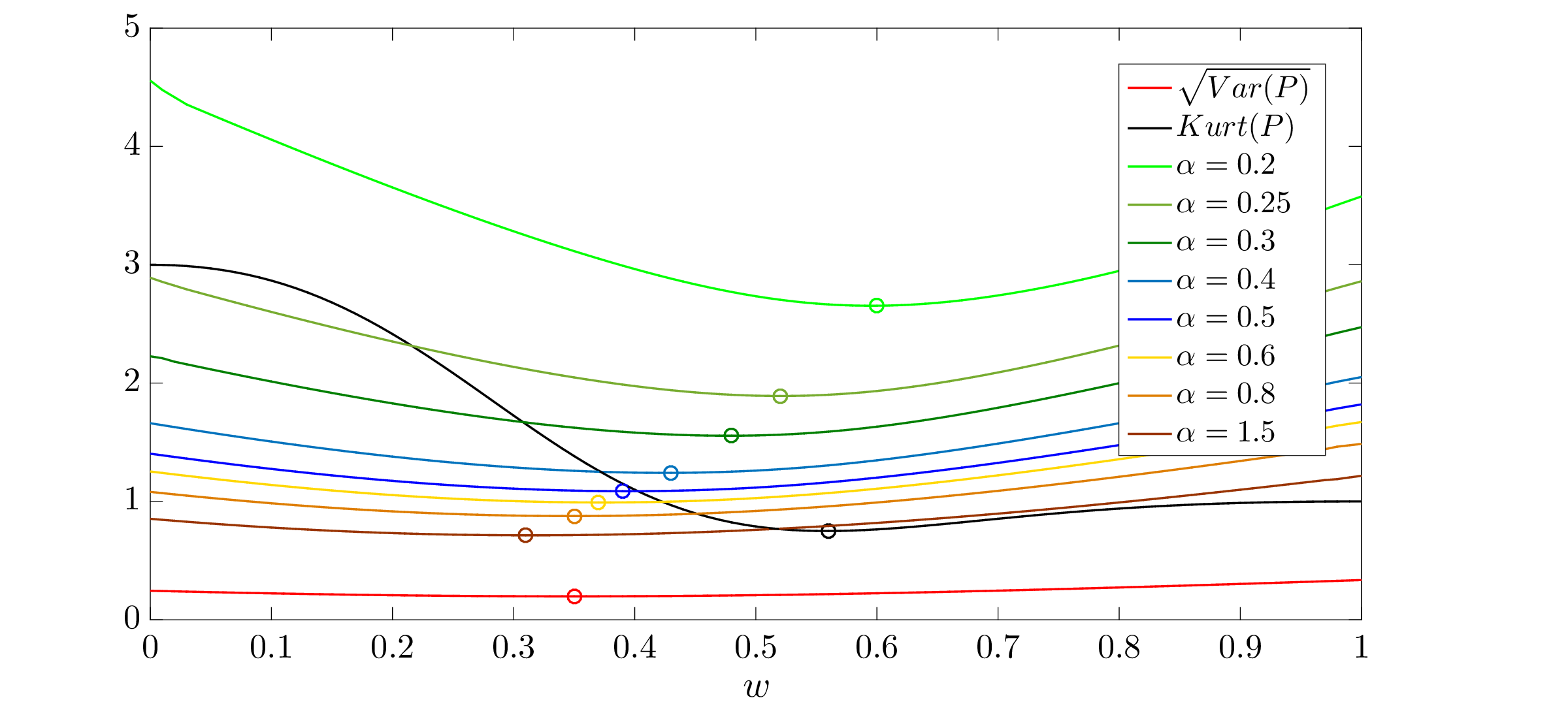}
\label{EffectOfKurtosis}
\vspace{-0.35cm}
\end{figure*}

Given that $k_2(\alpha)$ and $k_3(\alpha)$ are negative for all $\alpha$, the two additional terms $k_2(\alpha) \mathbb{S}\textup{kew}(P)^2$ and $k_3(\alpha) \mathbb{K}\textup{urt}(P)^2$ can be interpreted as driving the solution away from the Gaussian's skewness and kurtosis. This is intuitive as, under a fixed mean and variance, the Shannon entropy is maximized for the Gaussian distribution (\hyperref[Cover]{Cover and Thomas 2006}).

Finally, note that the optimization program in \eqref{ROptProgram} can accommodate an additional constraint on the portfolio expected return of the form $\mathbb{E}(P) = \boldsymbol{w}'\boldsymbol{\mu} \geq \mu_0$ with $\boldsymbol{\mu}$ the vector of assets' expected return, to account for the fact that investors do not only look at the risk, but also at the reward. In light of the GC expansion, such a framework would be linked to higher-moment efficient frontiers studied by e.g. \hyperref[Adcock]{Adcock (2014)} and \hyperref[Qi]{Qi et al. (2017)}. In the empirical study, we however concentrate on risk minimization due to the technical difficulties inherent in estimating the vector $\boldsymbol{\mu}$ (see Section \ref{Intro}) that yield significant loss in out-of-sample performance.

\section{Robust $m$-spacings estimator of $H_\alpha^{\exp}$}\label{Section4Est}
In this section, we explain how, given a finite sample $\{P_t\}_{1\leqslant t \leqslant T}$ with $P_t = \sum_{i=1}^n w_i X_{i,t}$, one can estimate $H_\alpha^{\exp}(P)$ in a robust way. In particular, we propose to use an estimator based on sample-spacings and discuss its properties in terms of consistency and robustness.


\subsection{Motivation and expression for the $m$-spacings estimator}\label{derivation}

To avoid making assumptions about the portfolio return's distribution, we are looking for a non-parametric estimator of $H_\alpha^{\exp}(P)$. There exists substantial research on non-parametric estimation of Shannon entropy, reviews of which can be found in \hyperref[Beirlant]{Beirlant et al. (1997)}. 

A natural way of estimating entropy is the plug-in estimate where a density estimator is plugged into the integral defining the entropy. One could for example choose the well-known Parzen (a.k.a. kernel) estimator. However, this estimator is known to be very sensitive to the bandwidth parameter, which can yield issues of stability for our portfolio optimization context.\footnote{When applied to our empirical data in Section \ref{Section5Empirical}, the Parzen estimator (with Gaussian kernel) achieves a worse risk-adjusted performance than the $m$-spacings estimator considered here for a wide range of values of the bandwidth parameter.} Instead, $m$-spacings estimation of entropy is more reliable: \hyperref[Wachowiak]{Wachowiak et al. (2005)} show that such estimators ``are robust and accurate, compare favorably to the popular Parzen window method for estimating entropies, and, in many cases, require fewer computations than Parzen methods.'' Therefore, we rely on a robust $m$-spacings estimator of Rényi entropy that extends the Shannon entropy $m$-spacings estimator of \hyperref[Miller]{Learned-Miller and Fisher (1993)}, a ``consistent, rapidly converging and computationally efficient estimator of entropy which is robust to outliers." 

Appendix \ref{DerivMSpacings} gives a detailed derivation of the estimator, of which we only report the final expression here for conciseness. 

\begin{prop}
Let $X$ be a continuous random variable. Then, the $m$-spacings estimator of $H_\alpha^{\exp}(X)$, denoted $\widehat{H}_\alpha^{\exp}(m,T)$, reads as
\begin{equation}\label{mspacingoriginal}
\begin{aligned}
&\widehat{H}_\alpha^{\exp}(m,T) :=  \\ &\Bigg(\frac{1}{T-m}  \sum_{i=1}^{T-m}  \bigg(\frac{T+1}{m}\big(X^{(i+m:T)}-X^{(i:T)}\big)\bigg)^{1-\alpha}\Bigg)^{\frac{1}{1-\alpha}}\;,
\end{aligned}
\end{equation}
where $X^{(1:T)} \leqslant X^{(2:T)} \leqslant \cdots \leqslant X^{(T:T)}$ are the order statistics of $X$ (i.e. the observations sorted by increasing order) and $m \in [1,T-1]$ is an integer parameter.
\end{prop}
\begin{proof}
See Appendix \ref{DerivMSpacings}.
\end{proof}
The parameter $m$ is a free parameter of great importance: increasing its value reduces the estimator's variance by grouping more order statistics in each spacing $X^{(i+m:T)}-X^{(i:T)}$. As a consequence, it plays a crucial role as it determines the robustness of the estimator, and in turn of the MRE portfolio. We come back to this in Section \ref{Robustness}.

Taking the limit where $\alpha\rightarrow 1$, we recover the exponential of the estimator of \hyperref[Miller]{Learned-Miller and Fisher (1993)}:
\begin{equation}\label{ExpShannonEst}
\begin{aligned}
&\widehat{H}_1^{\exp}(m,T) := \\ &\exp\Bigg(\frac{1}{T-m} \sum_{i=1}^{T-m}\ln\bigg(\frac{T+1}{m}\big(X^{(i+m:T)}-X^{(i:T)}\big)\bigg)\Bigg)\;.
\end{aligned}
\end{equation} 

\subsection{Properties of the $m$-spacings estimator}\label{PropertiesEst}
$m$-spacings estimation of entropy has attracted numerous research (see \hyperref[Beirlant]{Beirlant et al. 1997}) and dates a while back, e.g. \hyperref[Vasicek]{Vasicek (1976)}. However, it has been considered mainly for Shannon entropy  and, even for this specific case, only asymptotic behaviour has been studied. In the general case $\alpha \neq 1$, consistency has not been established. In this section, we first discuss the estimator's properties in terms of consistency, arguing that the estimator's asymptotic bias can be ignored for the sake of our portfolio application. Second, we show how the parameter $m$ determines the estimator's robustness.

\subsubsection{Asymptotic bias}\label{Consistency}

Let us first consider the case $\alpha=1$ and denote $\widehat{H}_1(m,T):=\ln\widehat{H}_1^{\exp}(m,T)$. \hyperref[VanEs]{van Es (1992)} proved that $\widehat{H}_1(m,T)$ is asymptotically biased but, interestingly, that the bias only depends on the fixed value of $m$ and not on the density $f_X$:
\begin{equation}\label{BiasEst}
\widehat{H}_1(m,T) - H_1(X) \rightarrow \psi(m) - \ln m \ \ \ \text{a.s.}\;,
\end{equation}
where $\psi(x)=\frac{d}{dx}\ln\Gamma(x)$ is the digamma function. 
Equation \eqref{BiasEst} means that we can simply subtract the bias to get a consistent estimator. Getting back to the exponential case, this means that
\begin{equation}\label{BiasCorrectedEst}
me^{-\psi(m)}\widehat{H}_1^{\exp}(m,T) \rightarrow H_1^{\exp}(X) \ \ \ \text{a.s.}
\end{equation}

Interestingly, as readily seen from \eqref{BiasCorrectedEst}, because the asymptotic bias depends only on $m$ when $\alpha=1$, using the bias-corrected estimator in \eqref{BiasCorrectedEst} or the biased estimator in \eqref{ExpShannonEst} is equivalent when searching the weights that minimize the entropy in \eqref{ROptProgram}. Ideally, we would want the same result to hold for all $\alpha$, i.e. the asymptotic bias to depend only on $\alpha$ and $m$. While such a result is not known, \hyperref[Hegde]{Hegde et al. (2005)} note that ``in many practical applications, [...]  this bias does not affect the solution, since it is independent of the true data distribution [...].'' Further, as we now show, the estimator's bias for $X$ and $\tilde{X} = (X-\mu)/\sigma $ is the same, i.e. it does not depend on the specific location and scale of $X$.

\begin{prop}\label{BiasIndep}
Let $\tilde{X} = (X-\mu)/\sigma$ and $\widehat{H}_\alpha(X;m,T) := \ln\widehat{H}^{\exp}_\alpha(X;m,T)$, then the $m$-spacings estimator's bias $\mathbb{B}\big(\widehat{H}_\alpha(X;m,T)\big) := \mathbb{E}\big(\widehat{H}_\alpha(X;m,T)\big) - H_\alpha(X) =  \mathbb{B}\big(\widehat{H}_\alpha(\tilde{X};m,T)\big)$.
\end{prop}
\begin{proof}
See Appendix \ref{AppendixBias}.
\end{proof}

\subsubsection{Robustness to outliers}\label{Robustness}

The parameter $m$ acts as a smoothing parameter that controls the estimator's variance. This section shows that increasing $m$ makes the $m$-spacings estimator more robust to outliers, which is crucial to ensure a solid out-of-sample performance of the MRE portfolio. Robustness conveys that a small perturbation from the true return distribution yields only a small change in the estimator's value. 

In assessing the robustness of an estimator, the \textit{Empirical Influence Function} (EIF) represents a useful tool (see \hyperref[Hampel]{Hampel et al. 1986}). Given an estimator $\hat{\theta}(X_1,\dots,X_T)$ of a quantity $\theta$ based on a sample  of size $T$, $\text{EIF}_{\hat{\theta}}(\hat{r})$ measures the sensitivity of the estimator $\hat{\theta}$ to the addition of a supplementary observation $\hat{r}$ in the sample:
\begin{equation}\label{EIFdef}
\text{EIF}_{\hat{\theta}}(\hat{r}) := (T+1)\big(\hat{\theta}(X_1,\dots,X_T;\hat{r})-\hat{\theta}(X_1,\dots,X_T)\big)\;.
\end{equation}
Intuitively, the lower is $\text{EIF}_{\hat{\theta}}(\hat{r})$, the more robust is the estimator $\hat{\theta}$. Figure \ref{EIFGaussian} illustrates the EIF of the $m$-spacings estimator - $\text{EIF}_{\widehat{H}^{\exp}_\alpha}(\hat{r})$ - for $T=250$ values from $X \sim \mathcal{N}(\mu=0, \sigma=0.2)$. Following \hyperref[Hampel]{Hampel et al. (1986)}, we set $X_i = \mu + \sigma \Phi^{-1}\big(\frac{i}{T+1}\big)$ to eliminate the random sample variability. We consider $\hat{r}\in [-5\sigma,5\sigma]$ and report the results for $\alpha=0.5$ only as other values yield similar insights. One can indeed observe that the EIF decreases with $m$ for large enough values of $\hat{r}$, i.e. for outliers.

\begin{figure*}[ht]
\caption{Increasing $m$ improves the robustness to outliers of the $m$-spacings estimator $\widehat{H}^{\exp}_\alpha(m,T)$. \textit{Notes:} $T=250$ observations are generated from $X \sim$ $\mathcal{N}(\mu=0, \sigma=0.2)$ by setting $X_i = \mu + \sigma \Phi^{-1}\big(\frac{i}{T+1}\big)$. The EIF of $\widehat{H}^{\exp}_\alpha(m,T)$ \-- $\text{EIF}_{\widehat{H}^{\exp}_\alpha}(\hat{r})$ \-- is then reported for $\alpha=0.5$ and different values of $m$. It decreases with $m$ in the space of outliers.}
\vspace{-0.3cm}
\hspace{2.5cm}\includegraphics[width=0.8\textwidth]{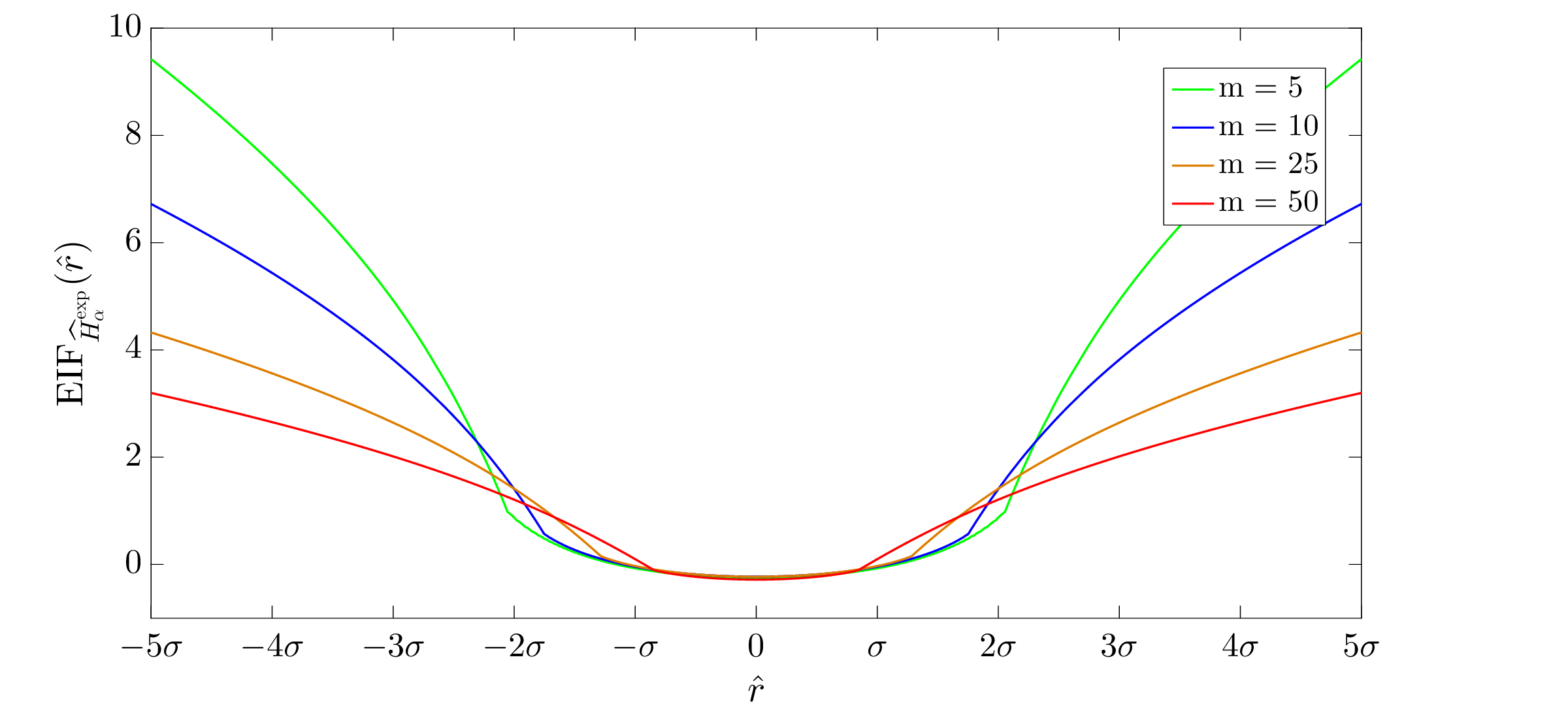}
\label{EIFGaussian}
\vspace{-0.2cm}
\end{figure*}

\section{Out-of-sample empirical study}\label{Section5Empirical}
We finish the article with an out-of-sample performance study of the MRE portfolio that aims at showing the practical interest of the proposed portfolio policy compared to several existing strategies. The study is performed on six datasets commonly used as benchmarks in the portfolio optimization literature.

\subsection{Methodology}\label{Methodo}

\subsubsection{Strategies of comparison}\label{ComparisonOtherModels}

The reported results compare the MRE portfolio with $\alpha \in \{ 0.3, 0.5, 0.7, 1, 1.5, 2 \}$ to five different minimum variance (MV) portfolios. The first four solve the quadratic optimization program 
\begin{equation}
\boldsymbol{w}^\star = \arg\underset{\boldsymbol{w}\in \mathcal{W}}{\min}\hspace{0.05cm} \boldsymbol{w}' \Sigma \boldsymbol{w}\;
\end{equation}
by estimating $\Sigma$ with the sample covariance matrix $\widehat{\Sigma}$ and the three robust shrinkage estimators developed by \hyperref[Ledoit1]{Ledoit and Wolf (2003, 2004a,b)}:
\begin{equation}\label{Shrinkage}
\begin{aligned}
&\widehat{\Sigma}_{CC} := \delta^\star \widehat{F}_{CC} + (1-\delta^\star)\widehat{\Sigma} \;,
\widehat{\Sigma}_{SF} := \delta^\star \widehat{F}_{SF} + (1-\delta^\star)\widehat{\Sigma} \;, \\ 
&\widehat{\Sigma}_{I} := \delta^\star \widehat{F}_{I} + (1-\delta^\star)\widehat{\Sigma}\;,
\end{aligned}
\end{equation} 
where $\delta^\star$ minimizes the Frobenius norm between the shrinkage estimator and the true matrix $\Sigma$. The three target matrices are based upon a constant correlation model ($\widehat{F}_{CC}$), a single-factor model ($\widehat{F}_{SF}$) and a multiple of the identity matrix ($\widehat{F}_{I}$). 

The fifth MV portfolio is the one-step M-portfolio (MP) of \hyperref[DeMiguelNogales]{DeMiguel and Nogales (2009)}:
\begin{equation}
(\boldsymbol{w}^\star,\mu^\star) = \arg \hspace{-0.12cm}\min_{\boldsymbol{w}\in \mathcal{W}, \hspace{0.05cm} \mu} \frac{1}{T}\sum_{t=1}^T \rho(P_t-\mu)\;,
\end{equation}
where $\rho$ is the Huber's robust loss function
\begin{equation}
\rho(x) := \left\{
                \begin{array}{ll}
                  x^2/2 \ \ \ \ \ \ \ \ \ \ \ \ \ \hspace{0.03cm} \text{if} \ |x|\leqslant c \\
                  c(|x|-c/2) \ \  \text{if} \ |x| > c  
                \end{array}\;, \ \ c=1\%\;.
              \right.
\end{equation}

We note that we have also implemented the robust Bayes-Stein mean-variance portfolio of \hyperref[Jorion]{Jorion (1986)} as well as the minimum VaR portfolio using the robust estimator in \hyperref[BoudtPeterson]{Boudt et al. (2008)}. We do not report their results as, even though such criteria are positively affected by higher returns, they feature a lower risk-adjusted performance than the MRE portfolio due to their sensitivity to the portfolio expected return. The equally-weighted strategy has been considered as well but, while it naturally achieves the lowest turnover, it is largely outperformed by all the other strategies and so is not reported either. 

\subsubsection{Datasets}\label{DataDescription}

We rely upon six monthly returns datasets from the Kenneth French library that are extensively used as benchmarks in the literature to compare portfolio strategies (see e.g. \hyperref[DeMiguel]{DeMiguel et al. 2009a,b}, \hyperref[Behr]{Behr et al. 2013} and \hyperref[Ardia]{Ardia et al. 2017}). The datasets are listed on Table \ref{ListBenchmark}.

\begin{table*}
\centering
\caption{List of datasets considered in the empirical study. \textit{Notes:} All datasets are made of monthly returns. The value weighting scheme is considered for the industry portfolios. Source: Kenneth French library.}
\label{ListBenchmark}
\vspace{-0.2cm}
\begin{tabular}{@{\hskip 0.15cm}l@{\hskip 0.3cm}c@{\hskip 0.35cm}cc@{\hskip 0.15cm}}
\hline \\[-0.35cm]
Datasets                                                                & Abb.   & $n$ & Time period           \\ \hline \\[-0.3cm]
6 Fama-French portfolios of firms sorted by size and book-to-market  & \textit{6BTM}   & 6   & 07/1963 - 06/2016  \\
25 Fama-French portfolios of firms sorted by size and book-to-market & \textit{25BTM}  & 25  & 07/1963 - 06/2016 \\
6 Fama-French portfolios of firms sorted by size and momentum  & \textit{6Mom}   & 6   & 07/1963 - 06/2016  \\
25 Fama-French portfolios of firms sorted by size and momentum & \textit{25Mom}  & 25  & 07/1963 - 06/2016 \\
10 industry portfolios representing the US stock market                 & \textit{10Ind} & 10  & 07/1963 - 06/2016 \\
17 industry portfolios representing the US stock market                 & \textit{17Ind} & 17  & 07/1963 - 06/2016 \\ \hline
\end{tabular}
\vspace{0cm}
\end{table*}

\subsubsection{Dynamic rebalancing}
We construct the portfolios by dynamic rebalancing. Rebalancing the weights not too frequently is important to ensure a satisfactory performance and turnover (\hyperref[Carroll]{Carroll et al. 2017}), hence we set a rolling window of one year as in \hyperref[Behr]{Behr et al. (2013)}. The estimation window is set to ten years, i.e. $T=120$. We use as starting date 07/1963 as in \hyperref[DeMiguel1]{DeMiguel et al. (2009b)} and \hyperref[Behr]{Behr et al. (2013)}, and rebalance the portfolios until 06/2016. This represents an out-of-sample period of 43 years.

\subsubsection{Optimization}\label{Optimization}

As pointed out in Section \ref{DefinitionROpt}, the MRE optimization program is, generally speaking, not a convex problem. Therefore, to find the optimal weights, we rely on the global optimizer of \hyperref[Ugray]{Ugray et al. (2007)} based on the Nelder-Mead algorithm. By doing so, we minimize the risk of getting stuck in a local minimum. We observe on the different datasets that recompiling the optimization several times yields essentially indistinguishable solutions, outlining that non-convexity is not a major issue.

\subsubsection{Choice of $m$}

As pointed out in Section \ref{Robustness}, the value of $m$ for the $m$-spacings estimator is of paramount importance as it determines the robustness of the MRE portfolio. To illustrate this, Table \ref{TurnoverPlot} reports, for the \textit{10Ind} dataset, the time evolution of the MRE optimal weights for $\alpha=0.5$ and $m=2,8,24$. These are unconstrained weights, i.e. corresponding to $\mathcal{W}=\{ \boldsymbol{w}\in \mathbb{R}^n \ | \ \boldsymbol{1}_n'\boldsymbol{w}=1 \}$. The weights of the M-portfolio are also reported for comparison. One can clearly observe that increasing $m$ improves the stability of the optimal weights obtained.

\begin{figure*}[ht]
\centering
\caption{Increasing $m$ improves the stability of the MRE optimal weights. \textit{Notes:} We report, for the \textit{10Ind} dataset, the time evolution of optimal weights for the M-portfolio and MRE portfolio with $\alpha=0.5$ and $m=2,8,24$. The weights are unconstrained, i.e. $\mathcal{W}=\{ \boldsymbol{w}\in \mathbb{R}^n \ | \ \boldsymbol{1}_n'\boldsymbol{w}=1 \}$.}
\vspace{-0.3cm}
\hspace{1.45cm}\includegraphics[width=0.9\textwidth]{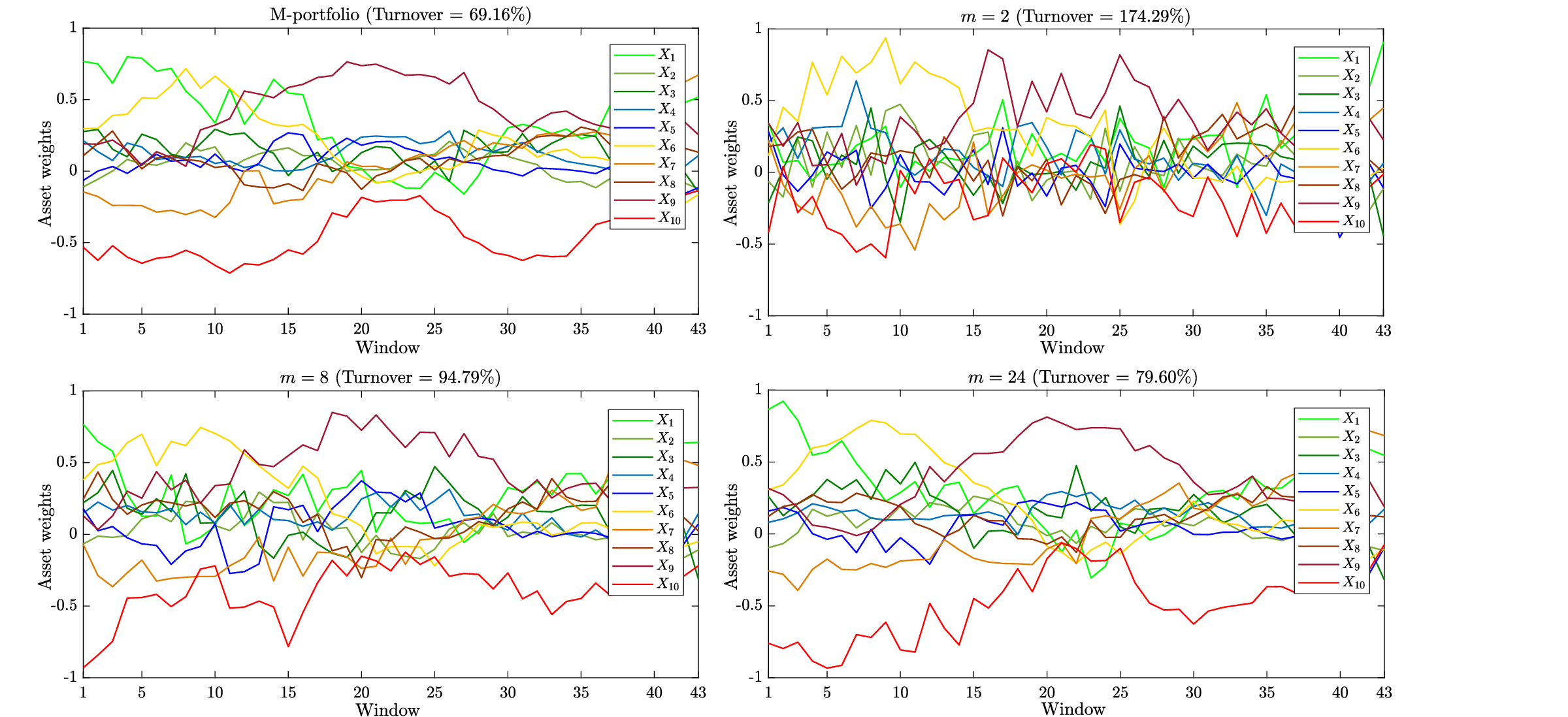}
\label{TurnoverPlot}
\vspace{0cm}
\end{figure*}

As a first strategy, we have considered the \textit{leave-one-out cross validation} method, using as criteria maximum return, minimum variance and maximum Sharpe ratio. However, the results obtained were quite poor both in terms of performance and turnover. Allowing $m$ to change at each rolling window seems to add instability to the procedure and is not recommended.

Therefore, as a second strategy, we have used the simple rule-of-thumb $m=\big[T^{2/3}\big]=24$ which works well on the considered datasets. We observe actually that once $m$ is high enough, the results display only a very minor sensitivity to the specific value of $m$ that is chosen, so that one can set $m=24$ without fearing that another different but close value would yield dissimilar results.\footnote{Specifically, the results in Table \ref{TableResultsPortfolios} for $m=18$ and $m=35$ yield a very similar performance. The only changes are in terms of turnover, which is higher for $m=18$ and nearly identical for $m=35$.}

\subsubsection{Weight constraints}\label{WConst}

To alleviate the estimation error, it is common in portfolio optimization to restrict the solution space $\mathcal{W}$. This has the effect of improving the stability of the optimal weights obtained and in turn the portfolios' out-of-sample performance. This is easily understood from Table \ref{TurnoverPlot}, in which all the portfolios feature a significant turnover in the unconstrained case, even though $n=10$ is relatively low in comparison to the sample size $T=120$. 

Therefore, we optimize the different portfolios subject to a constraint on the weights. We implement the global
variance-based constraint (GVBC) devised by \hyperref[Levy]{Levy and Levy (2014)}, which reads as
\begin{equation}
\sum_{i=1}^n \bigg(w_i-\frac{1}{n}\bigg)^2\frac{\sigma_i}{\bar{\sigma}} \leq \delta\;,
\end{equation}
where $\sigma_i := \sqrt{\mathbb{V}\text{ar}(X_i)}$ and $\bar{\sigma}:=\frac{1}{n}\sum_{i=1}^n \sigma_i$. The underlying rationale of GVBC is ``to impose more stringent constraints on stocks with relatively high standard deviations, as the estimation errors for these stocks’ parameters, and hence the potential economic loss, are larger than for stocks with relatively low standard deviations'' (\hyperref[Levy]{Levy and Levy 2014 p.375}). Using a U.S. industry portfolio dataset, the authors observe largely improved out-of-sample Sharpe ratios compared to several robust portfolio selection strategies for a wide range of values of $\delta$. In particular, the results are stable for $\delta$ between 10\% and 25\%. In the sequel, we set $\delta$ at the higher hand, i.e. $\delta=25\%$, as too low values make it difficult to distinguish between the different portfolios as they are too close to the equally-weighted one (corresponding to $\delta = 0$). The conclusions of the empirical study remain the same for $\delta=20\%$ and $\delta=15\%$, though naturally less strikingly.


\subsubsection{Performance measures}

We measure the portfolios' out-of-sample performance and stability with three criteria:
\begin{enumerate}
\item The Sharpe ratio, defined as
\begin{equation}
SR:=(\mathbb{E}(P)-r_f)/\sqrt{\mathbb{V}\text{ar}(P)}
\end{equation}
and estimated using sample estimators, which is the most common performance measure used in the asset allocation literature. For simplicity, we assume that $r_f=0$, as in e.g. \hyperref[DeMiguel1]{DeMiguel et al. (2009b)}, i.e. we report the reciprocal of the coefficient of variation.
\item Given that the appeal of the MRE portfolio compared to the minimum variance one is to account for higher-order uncertainty, using the Sharpe ratio alone is not sufficient to assess the merit of our portfolio policy. Hence, we also report the adjusted Sharpe ratio of \hyperref[Pezier]{Pézier (2004)} that accounts for investors' higher-moment preferences, defined as
\begin{equation}
ASR := SR\bigg(1+\frac{\mathbb{S}\text{kew}(P)}{3!}SR - \frac{\mathbb{K}\text{urt}(P)}{4!}SR^2\bigg)\;,
\end{equation}
that we estimate using sample moment estimators. 
\item To assess the stability and associated transaction costs of the portfolios, we report the turnover, defined as usual as
\begin{equation}
\text{Turnover} := \frac{1}{R-1} \sum_{t=1}^R \sum_{i=1}^n |w_{i,t+1}-w_{i,t^+}|\;,
\end{equation}
where $R=43$ is the number of rebalancing periods, $w_{i,t+1}$ is the desired weight of asset $i$ at time $t+1$ and $w_{i,t^+}$ is its weight before rebalancing at $t+1$.
\end{enumerate}
All three criteria are expressed in annual terms.

\subsection{Out-of-sample results}\label{EmpiricalResults}

\begin{table*}[ht]
\small
\centering
\caption{Out-of-sample Sharpe ratios, adjusted Sharpe ratios and turnover of the MRE and MV portfolios on the six datasets. \textit{Notes:} The portfolios are constructed with the methodology presented in Section \ref{Methodo}.}
\label{TableResultsPortfolios}
\vspace{-0.15cm}
\begin{tabular}{@{\hskip 0.08cm}l@{\hskip 0.50cm}c@{\hskip 0.1cm}c@{\hskip 0.1cm}c@{\hskip 0.15cm}c@{\hskip 0.15cm}c@{\hskip 0.15cm}c@{\hskip 0.65cm}c@{\hskip 0.30cm}c@{\hskip 0.30cm}c@{\hskip 0.30cm}c@{\hskip 0.30cm}c@{\hskip 0.08cm}}
\hline \\[-0.3cm]
\multicolumn{11}{c}{\textit{Sharpe ratio (SR)}}                                                                             \\
               & \multicolumn{6}{c}{\hspace{-0.45cm}\textit{MRE portfolios}}                                                     & \multicolumn{5}{c}{\hspace{-0.15cm}\textit{MV portfolios}}                                                                                                                                                  \\
               & $\alpha = 0.3$ & $\alpha = 0.5$ & $\alpha = 0.7$ & $\alpha = 1$ & $\alpha = 1.5$ & $\alpha = 2$ &  $\widehat{\Sigma}$ &$\widehat{\Sigma}_{CC}$   & $\widehat{\Sigma}_{SF}$   & $\widehat{\Sigma}_{I}$  & MP    \\
\textit{6BTM}  & 0.844          & 0.845          & 0.846          & 0.841        & 0.832          & 0.815  & 0.835     & 0.821                                                      & 0.833                                                      & 0.836                                                     & 0.841 \\
\textit{25BTM} & 0.985          & 0.980          & 0.973          & 0.961        & 0.937          & 0.906     & 0.953   & 0.913                                                      & 0.951                                                      & 0.957                                                     & 0.956 \\
\textit{6Mom}  & 0.768          & 0.763          & 0.753          & 0.750        & 0.716          & 0.700  & 0.738      & 0.741                                                      & 0.738                                                      & 0.737                                                     & 0.731 \\
\textit{25Mom} & 0.936          & 0.948          & 0.957          & 0.960        & 0.954          & 0.928   & 0.902     & 0.920                                                      & 0.904                                                      & 0.903                                                     & 0.916 \\
\textit{10Ind} & 0.995          & 1.001          & 1.014          & 1.013        & 0.971          & 0.936      & 0.977  & 0.970                                                      & 0.983                                                      & 0.973                                                     & 0.976 \\
\textit{17Ind} & 0.938          & 0.947          & 0.936          & 0.965        & 0.905          & 0.891    & 0.936    & 0.924                                                      & 0.935                                                      & 0.935                                                     & 0.918 \\ \textit{\textbf{Average}}   & \textbf{0.911} & \textbf{0.914} & \textbf{0.913} & \textbf{0.915} & \textbf{0.886} & \textbf{0.863} & \textbf{0.890}& \textbf{0.882}                                             & \textbf{0.891}                                             & \textbf{0.890}                                            & \textbf{0.890} \\ \hline \\[-0.3cm]
\multicolumn{11}{c}{\textit{Adjusted Sharpe ratio (ASR)}}                                                                                                                                                                                                                                                      \\
                & \multicolumn{6}{c}{\hspace{-0.45cm}\textit{MRE portfolios}}                                                     & \multicolumn{5}{c}{\hspace{-0.15cm}\textit{MV portfolios}}                                                                                                                                            \\
               & $\alpha = 0.3$ & $\alpha = 0.5$ & $\alpha = 0.7$ & $\alpha = 1$ & $\alpha = 1.5$ & $\alpha = 2$ & $\widehat{\Sigma}$  &$\widehat{\Sigma}_{CC}$   & $\widehat{\Sigma}_{SF}$   & $\widehat{\Sigma}_{I}$  & MP    \\
\textit{6BTM}  & 0.829          & 0.830          & 0.831          & 0.826        & 0.816          & 0.800    & 0.822    & 0.809                                                      & 0.821                                                      & 0.824                                                     & 0.826 \\
\textit{25BTM} & 0.969          & 0.963          & 0.956          & 0.943        & 0.919          & 0.889     & 0.940   & 0.906                                                      & 0.939                                                      & 0.944                                                     & 0.940 \\
\textit{6Mom}  & 0.759          & 0.753          & 0.744          & 0.741        & 0.708          & 0.694    & 0.730    & 0.733                                                      & 0.729                                                      & 0.729                                                     & 0.723 \\
\textit{25Mom} & 0.921          & 0.934          & 0.943          & 0.947        & 0.943          & 0.918    & 0.889    & 0.909                                                      & 0.892                                                      & 0.890                                                     & 0.902 \\
\textit{10Ind} & 0.990          & 0.995          & 1.005          & 1.004        & 0.965          & 0.927    & 0.973    & 0.966                                                      & 0.979                                                      & 0.968                                                     & 0.972 \\
\textit{17Ind} & 0.950          & 0.959          & 0.945          & 0.975        & 0.911          & 0.901  & 0.950      & 0.940                                                      & 0.950                                                      & 0.949                                                     & 0.929 \\ \textit{\textbf{Average}}   & \textbf{0.903} & \textbf{0.906} & \textbf{0.904} & \textbf{0.906} & \textbf{0.877} & \textbf{0.855} & \textbf{0.884} & \textbf{0.877}                                             & \textbf{0.885}                                             & \textbf{0.884}                                            & \textbf{0.882} \\ \hline \\[-0.3cm]
\multicolumn{11}{c}{\textit{Turnover}}                                                                                                                                                                                                                                                                         \\
               &  \multicolumn{6}{c}{\hspace{-0.45cm}\textit{MRE portfolios}}                                                     & \multicolumn{5}{c}{\hspace{-0.15cm}\textit{MV portfolios}}                                                                                                                                         \\
               & $\alpha = 0.3$ & $\alpha = 0.5$ & $\alpha = 0.7$ & $\alpha = 1$ & $\alpha = 1.5$ & $\alpha = 2$ & $\widehat{\Sigma}$ &$\widehat{\Sigma}_{CC}$   & $\widehat{\Sigma}_{SF}$   & $\widehat{\Sigma}_{I}$  & MP    \\
\textit{6BTM}  & 0.184          & 0.182          & 0.183          & 0.189        & 0.218          & 0.266    & 0.171    & 0.167                                                      & 0.170                                                      & 0.168                                                     & 0.168 \\
\textit{25BTM} & 0.499          & 0.535          & 0.616          & 0.966        & 1.356          & 1.502   & 0.448     & 0.464                                                      & 0.443                                                      & 0.445                                                     & 0.478 \\
\textit{6Mom}  & 0.167          & 0.171          & 0.180          & 0.191        & 0.247          & 0.284  & 0.168      & 0.168                                                      & 0.167                                                      & 0.168                                                     & 0.178 \\
\textit{25Mom} & 0.437          & 0.444          & 0.523          & 0.635        & 0.860          & 1.200    & 0.417    & 0.413                                                      & 0.410                                                      & 0.416                                                     & 0.422 \\
\textit{10Ind} & 0.348          & 0.350          & 0.378          & 0.450        & 0.600          & 0.776   & 0.283     & 0.276                                                      & 0.274                                                      & 0.286                                                     & 0.321 \\
\textit{17Ind} & 0.526          & 0.568          & 0.685          & 0.825        & 1.033          & 1.246    & 0.449    & 0.422                                                      & 0.433                                                      & 0.452                                                     & 0.467 \\ \textit{\textbf{Average}}   & \textbf{0.360} & \textbf{0.375} & \textbf{0.427} & \textbf{0.542} & \textbf{0.719} & \textbf{0.879} & \textbf{0.323} & \textbf{0.318}                                             & \textbf{0.316}                                             & \textbf{0.322}                                            & \textbf{0.339} \\ \hline
\end{tabular}
\vspace{0cm}
\end{table*}

The results are reported on Table \ref{TableResultsPortfolios}. Several interesting observations can be made. 

First, comparing the six MRE portfolios, one can clearly see that $\alpha=1.5$ and $\alpha=2$ yield by far the worst performance. This is consistent with the Gram-Charlier expansion in \eqref{GCApprox}, where setting $\alpha>1$ favors solutions with more kurtosis (as $k_1(\alpha)<0$ in that case). Therefore, both from a theoretical and empirical perspective, it is recommended to set $\alpha \leqslant 1$, as argued in Section \ref{AlphaZeroOne}.

Second, the turnover of the MRE portfolio systematically increases with $\alpha$. It does not increase too much from $\alpha=0.3$ to $\alpha=0.7$ but then quickly increases dramatically. This effect can be explained by the fact that the convexity of $H_\alpha^{\exp}(P)$ (as a function of $\boldsymbol{w}$) decreases as $\alpha$ increases. This is for example visible in Figure \ref{EffectOfKurtosis}. This makes that the minimum is more bound to change largely from one rolling window to another as the objective function is flatter around the minimum the higher the $\alpha$.

Third, it is appealing to observe that, for $\alpha\in\{0.3,0.5,0.7,1\}$, the performance of the MRE portfolio is very stable with respect to $\alpha$. This is easily observed by looking at the average $SR$ and $ASR$ across the six datasets. This parameter robustness is an appealing behaviour for the decision-maker. Combined with the fact that the turnover increases with $\alpha$, this means that choosing $\alpha$ quite low, in this case around $\alpha=0.3$, yields the best trade-off between risk, return and turnover.\footnote{For completeness, we have checked the results for $\alpha\in\{0.05,0.1,0.2\}$ as well. The turnover barely decreases compared to $\alpha=0.3$, and the Sharpe ratio measures remain nearly identical.}

Fourth, comparing the MRE and MV portfolios, one can definitely observe that the MRE portfolios improve upon the MV portfolios both in terms of $SR$ and $ASR$, except naturally for $\alpha=1.5$ and $\alpha=2$ that we discard in the discussion below. Indeed, averaging across the six datasets, each MRE portfolio displays larger $SR$ and $ASR$ than \textit{all} the MV portfolios. In terms of turnover however, all the MRE portfolios display less stability than the MV portfolios. This is to be expected as the MRE portfolio is sensitive to the higher-order moments, which are more affected by outliers than the variance. That said, for $\alpha=0.3$ and $\alpha=0.5$, the increase in turnover is quite modest (around 4 percentage points on average for $\alpha=0.3$).

Therefore, we conclude that Rényi entropy provides a better risk criterion than the variance in an asset allocation context, especially for low values of $\alpha$, specifically around $\alpha=0.3$ for the datasets considered here.

\section{Conclusion}\label{Conclusion}
Many studies from the wide scientific literature suggest that minimum risk portfolios exhibit solid out-of-sample performances in spite of the fact that there is no target return constraint. Whereas variance \-– initially introduced by Markowitz \-– is a natural risk measure in a Gaussian framework, it fails to capture extreme events that arguably arise in real applications. In order to take this reality into account, various alternative risk measures have been put forward. 

In this article, we have proposed a natural uncertainty measure \-- the exponential Rényi entropy \-- as a higher-moment criterion for portfolio selection. Rényi entropy generalizes Shannon entropy, yielding a set of uncertainty measures. Its parameter $\alpha$ enables to tune the relative contributions of the central and tail parts of the distribution in the measure. Its exponential transform fulfills desirable properties as it is closely related to the class of deviation risk measures, as well as to measures of spread for $\alpha\in[0,1]$. 

Minimizing this measure yields the minimum Rényi entropy portfolio. A Gram-Charlier expansion shows that this portfolio represents a higher-moment extension of the minimum variance portfolio, with $\alpha$ controlling the trade-off between variance and kurtosis minimization. 

In practical settings, the empirical study has demonstrated that the minimum Rényi entropy portfolio fares better out-of-sample compared to state-of-the-art robust minimum variance portfolios in terms of trading off risk, return and turnover, especially for $\alpha$ close to zero.

Beyond our application, this article points the appeal of Rényi entropy in various operations research problems as a powerful way of capturing higher-moment uncertainty, and of using entropy as an optimization criterion rather than just an ad hoc evaluation measure. In the particular case of portfolio selection, Rényi entropy has been shown to be a powerful alternative to existing risk criteria, opening the door to other applications. For instance, one may apply Rényi entropy to the \textit{risk parity} strategy, which raises the challenge of computing the assets' contributions to the portfolio return's exponential Rényi entropy. 

\section*{Acknowledgements}
The authors are grateful to Kris Boudt and Victor DeMiguel for stimulating discussions. The authors also thank participants of the \textit{Actuarial and Financial Mathematics 2018 Conference}, the \textit{35th Annual Conference of the French Finance Association (AFFI)} and the \textit{2018 Belgian Financial Research Forum} for their comments and suggestions. This work was supported by the Fonds National de la Recherche Scientifique (F.R.S.-FNRS) [grant number FC 17775].

\appendix
\section*{Appendices}
\renewcommand{\thesubsection}{\Alph{subsection}}
\renewcommand{\thetable}{\Alph{subsection}.\arabic{table}}
\setcounter{table}{0}
\renewcommand{\theequation}{\Alph{subsection}.\arabic{equation}}
\renewcommand{\theremark}{\Alph{subsection}.\arabic{remark}}
\numberwithin{prop}{subsection}
\numberwithin{equation}{subsection}
\numberwithin{remark}{subsection}
\counterwithin{figure}{subsection}
\counterwithin{table}{subsection}

\subsection{Counter-examples to sub-additivity of $H_\alpha^{\exp}$}\label{AppendixSubAdd}

In this section, we report the three counter-examples to sub-additivity as mentioned in Proposition \ref{PropSubAdd}.

\subsubsection{Lévy distributions}

\begin{prop*}
$H_1^{\exp}$ is not sub-additive for a pair $(X,Y)$ of independent Lévy-distributed random variables.
\end{prop*}
\begin{proof}
The pdf of $Z\sim\text{Lévy}(\mu,\sigma)$ is given by $f_Z(x)=\sqrt{\frac{\sigma}{2\pi}}\frac{e^{-\frac{\sigma}{2(x-\mu)}}}{(x-\mu)^{3/2}}$ and is strictly positive for $x>\mu$. Its exponential entropy is known in closed-form (see \hyperref[Zografos]{Zografos and Nadarajah 2003}): $H_1^{\exp}(Z) = 4\sigma\sqrt{\pi}e^{\frac{1+2\gamma}{2}}$, where $\gamma\approx 0.577$ is the Euler-Mascheroni constant. We note the parameters of $X,Y$ as $(\mu_X,\sigma_X)$ and $(\mu_Y,\sigma_Y)$, respectively, with $\sigma_X,\sigma_Y>0$. As Lévy is a stable law, the sum $X+Y$ is again Lévy-distributed with parameters
$\mu_{X + Y}=\mu_X + \mu_Y$ and  
$\sigma_{X+Y}=\sigma_X+\sigma_Y+2\sqrt{\sigma_X\sigma_Y}$. Sub-additivity is thus equivalent to $\sigma_{X+Y} \leqslant \sigma_X + \sigma_Y\Leftrightarrow 2\sqrt{\sigma_X}\sqrt{\sigma_Y} \leqslant 0$, which never holds when $\sigma_X,\sigma_Y > 0$. 
\end{proof}

\subsubsection{Bimodal distributions}

\begin{prop*}
Consider $(Z_X,Z_Y)$ a pair of independent standard Normal variables and $(U_X,U_Y)$ a pair of independent Bernoulli variables of parameter $1/2$, independent from both $Z_Y,Z_Y$. Define $X:=(2\mu_X U_X-\mu_X)+\sigma Z_X,Y:=(2\mu_Y U_Y-\mu_Y)+\sigma Z_Y$ with constants $\mu_X,\mu_Y$ and $\sigma>0$. Then, $H_1^{\exp}$ is not sub-additive for the pair $(X,Y)$ whenever e.g. $(\mu_X,\mu_Y)=(1,2)$ and $\sigma<0.3918$.
\end{prop*}
\begin{proof}
Noting $\phi(x)$ the standard Gaussian density, the marginal densities of $X,Y$ are given by the Gaussian mixtures
\begin{equation}\label{DensityXY}
\begin{aligned}
&f_X(x) = \frac{1}{2\sigma}\bigg(\phi\Big(\frac{x+\mu_X}{\sigma}\Big) + \phi\Big(\frac{x-\mu_X}{\sigma}\Big)\bigg) \;, \\ &f_Y(x) = \frac{1}{2\sigma}\bigg(\phi\Big(\frac{x+\mu_Y}{\sigma}\Big) + \phi\Big(\frac{x-\mu_Y}{\sigma}\Big)\bigg)\;.
\end{aligned}
\end{equation}
It is easy to show (see e.g. \hyperref[Pham]{Pham and Vrins 2005}) that the density of $X+Y$ is
\begin{equation}\label{DensityXPlusY}
\begin{aligned}
f_{X+Y}(x) = &\frac{1}{4\sigma\sqrt{2}}\bigg(\phi\Big(\frac{x+\mu_X+\mu_Y}{\sigma\sqrt{2}}\Big) + \phi\Big(\frac{x+\mu_X-\mu_Y}{\sigma\sqrt{2}}\Big) \\
&+ \phi\Big(\frac{x-\mu_X+\mu_Y}{\sigma\sqrt{2}}\Big) + \phi\Big(\frac{x-\mu_X-\mu_Y}{\sigma\sqrt{2}}\Big) \bigg)\;.
\end{aligned}
\end{equation}
Because $H_1$ only depends on the density $f_X$ of $X$, we denote $H_1[f_X]:=H_1(X)$.
\hyperref[Vrins]{Vrins et al. (2007)} show that, for a random variable $Z$ whose density can be written in the form $f_Z(x)=\sum_{n=1}^N \pi_n K_n(x)$ with positive weights $\pi_n$ summing to 1 and Gaussian kernels $K_n(x) = \frac{1}{\sigma_n}\phi\big(\frac{x-\mu_n}{\sigma_n}\big)$, then $H_1(Z)$ can be bounded below and above. More explicitly,
\begin{equation}
\underline{H}_1(Z) \leqslant H_1(Z) \leqslant \overline{H}_1(Z)\;,
\end{equation}
with 
\begin{equation}\label{FormulaBounds}
\begin{aligned}
&\overline{H}_1(Z) := \sum_{n=1}^N \pi_n H_1[K_n] + h(\boldsymbol{\pi}) \;, \\ 
&\underline{H}_1(Z) := \overline{H}_1(Z) - \sum_{n=1}^N \pi_n\epsilon_n' - \sum_{n=1}^N \pi_n\bigg[\ln\bigg(\frac{s}{\pi_n s_n}\bigg) +1\bigg]\epsilon_n\;,
\end{aligned}
\end{equation}
where $h(\boldsymbol{\pi}): =-\sum_{n=1}^N \pi_n \ln \pi_n$, $s_n:=\max_x K_n(x)=\big(\sqrt{2\pi}\sigma_n\big)^{-1}$ and $s:=\max_n s_n$. Rearranging the $\mu_n$ by increasing order and defining $d_n := \min(\mu_n-\mu_{n-1},\mu_{n+1}-\mu_n)$ with $\mu_0 :=-\infty$, $\mu_{N+1}:=\infty$ by convention, we have
\begin{equation}
\epsilon_n := \text{Erfc}\bigg(\frac{d_n}{2\sqrt{2}\sigma_n}\bigg)  \ \ , \ \ 
\epsilon_n' := \frac{1}{2}\epsilon_n + \frac{d_n}{2\sqrt{2\pi}\sigma_n}e^{-\big(\frac{d_n}{2\sqrt{2}\sigma_n}\big)^2}\;,
\end{equation}
with the complementary error function $\text{Erfc}(x):=2\Phi(-x\sqrt{2})$, where $\Phi$ is the standard Gaussian cdf.

Using these lower and upper bounds, the $H_1^{\exp}$ operator fails to be sub-additive for the pair $(X,Y)$ if
\begin{equation}\label{ConditionBounds}
\exp(\underline{H}_1(X+Y)) > \exp(\overline{H}_1(X)) + \exp(\overline{H}_1(Y))\;.
\end{equation}
Indeed, the LHS is a lower-bound to $H_1^{\exp}(X+Y)$ while the RHS is an upper-bound to $H_1^{\exp}(X)+H_1^{\exp}(Y)$. Setting $(\mu_X,\mu_Y)=(\mu,2\mu)$, the bounds in \eqref{FormulaBounds} applied to the densities in \eqref{DensityXY}-\eqref{DensityXPlusY} read as
\begin{equation}
\begin{aligned}
\overline{H}_1(X) = \overline{H}_1(&Y) = \ln\big(2\sigma\sqrt{2\pi e}\big) \;, \\
\underline{H}_1(X+Y) = &\ln\big(8\sigma\sqrt{\pi e}\big)-\text{Erfc}(\mu/2\sigma)(3/2+\ln(4)) \\  &- \frac{\mu}{2\sigma\sqrt{\pi}}e^{-(\mu/2\sigma)^2}\;,
\end{aligned}
\end{equation}
from which we find for example that, setting $\mu=1$, \eqref{ConditionBounds} holds as long as $0<\sigma<0.3918$.
\end{proof}

\subsubsection{Comonotonic random variables}

\begin{prop*}
$H_1^{\exp}$ is not sub-additive for a comonotonic pair $(X,Y)$, with $Y=F(X)$, $F'(X)\sim Exp(1)$. 
\end{prop*}
\begin{proof}
Two random variables $X,Y$ are comonotonic when $Y$ can be written as $F(X)$ where $F$ is a continuous strictly increasing function. Denote $G(x) := x + F(x)$, which is also strictly increasing and so invertible, and denote $H(x)$ its inverse. Then, the cdf of $X+Y=G(X)$ is given by
\begin{equation*}
F_{X+Y}(x) = \mathbb{P}(G(X) \leqslant x) = \mathbb{P}(X \leqslant H(x))=F_X(H(x))\;,
\end{equation*}
and its pdf reads
\begin{equation*}
f_{X+Y}(x) = \frac{f_X(H(x))}{G'(H(x))}\;.
\end{equation*}
As a result, $H_1^{\exp}(X+Y)$ becomes
\begin{equation*}
H_1^{\exp}(X+Y) = \exp \Bigg(-\int \frac{f_X(H(x))}{G'(H(x))} \ln\frac{f_X(H(x))}{G'(H(x))} \hspace{0.05cm} dx\Bigg)\;.
\end{equation*}
A change of variable $z=H(x)$ and algebraic manipulations lead to
\begin{equation*}
H_1^{\exp}(X+Y) = H_1^{\exp}(X) \hspace{0.05cm} \exp \big( \mathbb{E}\big(\ln (1+F'(X))\big)\big)\;.
\end{equation*}
Based on a similar reasoning, we can show that
\begin{equation*}
H_1^{\exp}(Y) = H_1^{\exp}(X) \hspace{0.05cm} \exp \big( \mathbb{E}\big(\ln F'(X)\big)\big)\;,
\end{equation*}
meaning that sub-additivity amounts to showing that 
\begin{equation}\label{additivitycondition}
\exp\big(\mathbb{E}\big(\ln(1+F'(X))\big)\big) \leqslant 1 + \exp\big(\mathbb{E}\big(\ln F'(X)\big)\big)\;.
\end{equation}

Let us now show instead a counter-example to \eqref{additivitycondition} where the left-hand side is higher than the right-hand side, i.e. where $H_1^{\exp}$ is super-additive. Because $F'(X)$ has to be a positive random variable ($F$ is strictly increasing), we consider $F'(X):=\zeta\sim Exp(1)$. Define $W:=\ln(1+\zeta)$ and $Z:=\ln\zeta$. To have super-additivity, we have to show that 
\begin{equation}\label{additiveExp}
e^{\mathbb{E}(W)} > 1+e^{\mathbb{E}(Z)}\;.
\end{equation}
One can find the pdf of $W$ and $Z$ to be given by
\begin{gather}
\label{pdfW} f_W(x) = e^{1+x-e^x} \ , \ x \geqslant 0\;, \\
\label{pdfZ} f_Z(x) = e^{x-e^x} \ , \ x \in \mathbb{R} \;.
\end{gather}
We can now compute the expectations. From \eqref{pdfW}, $\mathbb{E}(W)$ is given by the following integral:
\begin{equation*}
\begin{aligned}
\mathbb{E}(W)
&= e \int_0^{\infty} xe^{x-e^x}dx\;.
\end{aligned}
\end{equation*}
A change of variable $z=e^x$ and integration by parts yields
\begin{equation*}
\mathbb{E}(W) = e \hspace{0.05cm} \Gamma(0,1) \approx 0.596\;,
\end{equation*}
where $ \Gamma(s,x)=\int_x^\infty t^{s-1}e^{-t}dt$ is the upper incomplete Gamma function. A similar derivation yields $\mathbb{E}(Z)=-\gamma \approx -0.577$, i.e. minus the Euler-Mascheroni constant. Finally, we have in agreement with \eqref{additiveExp} that
\begin{equation*}
e^{\mathbb{E}(W)} = e^{e \hspace{0.05cm} \Gamma(0,1)} \approx 1.815 > 1.561 \approx 1+e^{-\gamma} = 1+e^{\mathbb{E}(Z)}\;,
\end{equation*}
hence providing a counter-example to sub-additivity. 
\end{proof}

\subsection{Proofs of Propositions}



\subsubsection{Proof of Proposition \ref{SubAddElliptical}} \label{ProofElliptical}

\begin{proof}
We set $\boldsymbol{\mu}=\boldsymbol{0}$ without loss of generality as $H_\alpha^{\exp}$ is translation-invariant.  The proof relies on the special convolution properties of elliptical distributions, see e.g. \hyperref[Fang]{Fang and Zhang (1990)}.  In particular, any linear combination of an elliptical random vector remains elliptical, meaning that we can write
\begin{equation*}
X \sim \text{El}(\sigma_X^2,g_1) \ \ , \ \ 
Y \sim \text{El}(\sigma_Y^2,g_1)\ \ , \ \
X+Y \sim \text{El}(e'\Sigma e,g_1),
\end{equation*}
where $e=(1,1)'$, and so $e'\Sigma e = \sigma_X^2 + \Sigma_Y^2 + 2\rho\sigma_X\sigma_Y$. Moreover, any elliptical distribution $\boldsymbol{X}$ can be written as $\boldsymbol{X} = \boldsymbol{\mu} + A\boldsymbol{Y}$, where $AA'=\Sigma$ and $\boldsymbol{Y}$ is a spherical distribution, i.e. an elliptical distribution with $\Sigma = I$, the identity matrix. Applied to our case, this means that we can write $X = \sigma_X Z$, $Y = \sigma_Y Z$ and $X+Y = \sqrt{e'\Sigma e}Z$, with $Z \sim \text{El}(1,g_1)$. Finally, the sub-additivity of $H_\alpha^{\exp}$ reduces to 
\begin{equation*}
\begin{aligned}
&H_\alpha^{\exp}(X+Y) \leqslant H_\alpha^{\exp}(X) + H_\alpha^{\exp}(Y) \\
\Leftrightarrow \hspace{0.15cm} &\sqrt{e'\Sigma e}H_\alpha^{\exp}(Z) \leqslant \sigma_X H_\alpha^{\exp}(Z) + \sigma_Y H_\alpha^{\exp}(Z) \\
\Leftrightarrow \hspace{0.15cm} &\sqrt{\sigma_X^2 + \sigma_Y^2 + 2\rho\sigma_X\sigma_Y} \leqslant \sigma_X + \sigma_Y \;,
\end{aligned}
\end{equation*}
which is true for any $\rho \in [-1,1]$.
\end{proof}

\subsubsection{Proof of Proposition \ref{PropGC}}\label{AppendixGC}

\begin{proof}
The truncated Gram-Charlier (GC) expansion of the pdf of $X'$ is given by
\begin{equation}
f_{X'}(x) \approx \phi(x)\bigg(1+\mathbb{S}\textup{kew}(X)\frac{H_3(x)}{3!}+\mathbb{K}\textup{urt}(X)\frac{H_4(x)}{4!}\bigg)\;.
\end{equation}
The proof relies on special properties of the Hermite polynomials $H_i$'s. Those are defined in relation with the derivatives of the standard Gaussian pdf $\phi$:
\begin{equation*}
\frac{\partial^i \phi(x)}{\partial x^i} = (-1)^i H_i(x) \phi(x)\;.
\end{equation*}
The first four polynomials are given by $H_1(x)=x$, $H_2(x)=x^2-1$, $H_3(x)=x^3-3x$ and $H_4(x)=x^4-6x^2+3$. They form an orthonormal system in the sense that
\begin{equation*}
 \int H_i(x)H_j(x)\phi(x)dx = \left\{
                \begin{array}{ll}
                  i! \ \ \text{if} \ i=j \\
                  0 \ \ \hspace{0.05cm} \text{if} \ i \neq j 
                  
                \end{array}\;.
              \right.
\end{equation*}
To find the GC expansion of $H_\alpha(X)$, we want to have a similar system for the $\alpha^{\text{th}}$ power of $\phi$. One can check that $\phi^\alpha(x) = k \phi(\sqrt{\alpha}x)$ with $k=(2\pi)^{\frac{1-\alpha}{2}}$ and that 
\begin{equation*}
\frac{\partial^i \phi(\sqrt{\alpha}x)}{\partial x^i} = (-1)^i \widetilde{H}_i(x)\phi(\sqrt{\alpha}x)\;,
\end{equation*}
with $\widetilde{H}_1(x)=\alpha x$, $\widetilde{H}_2(x)=\alpha^2 x^2 - \alpha$, $\widetilde{H}_3(x)=\alpha^3 x^3 - 3\alpha^2 x$ and $\widetilde{H}_4(x)=\alpha^4 x^4 - 6\alpha^3 x^2 + 3\alpha^2$. Hence, in relation with the original polynomials $H_i$'s, $\phi(\sqrt{\alpha}x)$ forms the system
\begin{equation*}
 \int H_i(x)H_j(x)\phi(\sqrt{\alpha}x)dx = \left\{
                \begin{array}{ll}
                  C_i(\alpha) \ \ \text{if} \ i=j \\
                  0 \ \ \ \ \ \ \  \hspace{0.11cm} \text{if} \ i \neq j 
                  
                \end{array}\;.
              \right.
\end{equation*}
Following algebraic manipulations, the first four coefficients $C_i(\alpha)$'s express as
\begin{equation*}
\begin{aligned}
C_1(\alpha) &= \alpha^{-3/2}\;, \\
C_2(\alpha) &= \frac{(\alpha-2)\alpha + 3}{\alpha^{5/2}}\;, \\
C_3(\alpha) &= \frac{3(3(\alpha-2)\alpha + 5}{\alpha^{7/2}}\;, \\
C_4(\alpha) &= \frac{3(3\alpha^4 - 12\alpha^3 + 42\alpha^2 - 60\alpha + 35)}{\alpha^{9/2}}\;.
\end{aligned}
\end{equation*}
Let us now first derive the GC expansion of $I_\alpha(X') := \int (f_{X'}(x))^\alpha dx$. Using the results above and the second-order Taylor expansion $(1+\varepsilon)^\alpha \approx 1+\alpha \varepsilon + \frac{\alpha(\alpha-1)}{2}\varepsilon^2$, $I_\alpha(X')$ approximates as
\begin{equation*}
\begin{aligned}
&I_\alpha(X') \approx I\big[\mathcal{N}(0,1)\big] + \frac{3k\alpha(\alpha-1)^2}{4!\alpha^{5/2}} \mathbb{K}\text{urt}(X) \\ &+\hspace{-0.08cm} \frac{k\alpha(\alpha-1)C_3(\alpha)}{2\times 3!^2}\mathbb{S}\text{kew}(X)^2 \hspace{-0.05cm}+\hspace{-0.08cm} \frac{k\alpha(\alpha-1)C_4(\alpha)}{2\times 4!^2} \mathbb{K}\text{urt}(X)^2.
\end{aligned}
\end{equation*}
Note that there is no $\mathbb{S}\text{kew}(X)$ term left because $\int \phi(\sqrt{\alpha}x)H_3(x)dx=0$. Now, to finish, we need to get back to $H_\alpha(X) = \frac{1}{1-\alpha}\ln I_\alpha(X)$. We apply the Taylor expansion
\begin{equation*}
\ln\big(I\big[\mathcal{N}(0,1)\big]+\varepsilon\big) \approx \ln I\big[\mathcal{N}(0,1)\big] + I\big[\mathcal{N}(0,1)\big]^{-1}\varepsilon\;,
\end{equation*}
where $I\big[\mathcal{N}(0,1)\big] = \sqrt{(2\pi)^{1-\alpha}/\alpha}$, which finally yields 
\begin{equation*}
\begin{aligned}
H_\alpha(X) \approx \hspace{0.1cm} &H_\alpha\big[\mathcal{N}(0,\mathbb{V}\textup{ar}(X))\big] + k_1(\alpha) \mathbb{K}\textup{urt}(X) \\
&+ k_2(\alpha) \mathbb{S}\textup{kew}(X)^2 + k_3(\alpha) \mathbb{K}\textup{urt}(X)^2\;,
\end{aligned}
\end{equation*}
where the functions $k_1(\alpha)$, $k_2(\alpha)$ and $k_3(\alpha)$ are made explicit in \eqref{CoeffGC}.
\end{proof}

\subsubsection{Proof of Proposition \ref{BiasIndep}}\label{AppendixBias}

\begin{proof}
For notation purposes, we denote $X_0 := X'$. From \hyperref[Royston]{Royston (1982)}, the density of $X^{(r:T)}$, the $r^{\text{th}}$ order statistics of $X$, writes as
\begin{equation}\label{DensityOrder}
f_{X^{(r:T)}}(x) = (1-F_X(x))^{r-1}(F_X(x))^{T-r}f_X(x)\;.
\end{equation}
As we have $F_X(x)=F_{X_0}\big(\frac{x-\mu}{\sigma}\big)$ and $f_X(x)=f_{X_0}\big(\frac{x-\mu}{\sigma}\big)/\sigma$, we find from \eqref{DensityOrder} that the density of $X^{(r:T)}$ is given by
\begin{equation*}
f_{X^{(r:T)}}(x) = \frac{1}{\sigma} f_{X_0^{(r:T)}}\Big(\frac{x-\mu}{\sigma}\Big)\;,
\end{equation*}
meaning that
\begin{equation}\label{LinkOrder}
X^{(r:T)} \sim \mu + \sigma X_0^{(r:T)}\;.
\end{equation}
Replacing \eqref{LinkOrder} in the expression of $\widehat{H}_\alpha(X;m,T)$, we can write 
\begin{equation}\label{Eq53}
\widehat{H}_\alpha(X;m,T) = \ln \sigma + \widehat{H}_\alpha(X_0;m,T)\;.
\end{equation}
Moreover, as $H_\alpha^{\exp}(X)=\sigma H_\alpha^{\exp}(X_0)$, we have $\ln \sigma = H_\alpha(X) - H_\alpha(X_0)$. Replacing this in \eqref{Eq53} yields 
\begin{equation*}\label{EqComp}
\begin{aligned}
& \widehat{H}_\alpha(X;m,T) =  H_\alpha(X) - H_\alpha(X_0) +  \widehat{H}_\alpha(X_0;m,T) \\
\Leftrightarrow \hspace{0.15cm} & \widehat{H}_\alpha(X;m,T) - H_\alpha(X) = \widehat{H}_\alpha(X_0;m,T) - H_\alpha(X_0) \\  
\Leftrightarrow \hspace{0.15cm} & \mathbb{B}\big(\widehat{H}_\alpha(X;m,T)\big) = \mathbb{B}\big(\widehat{H}_\alpha(X_0;m,T)\big)\;,
\end{aligned}
\end{equation*}
which completes the proof.
\end{proof}

\subsection{Derivation of the $m$-spacings estimator of $H_\alpha^{\exp}$}\label{DerivMSpacings}

This appendix derives the $m$-spacings estimator of $H_\alpha^{\exp}$, whose final expression is reported in Section 4.1.

Consider i.i.d. copies $X^1, X^2, \dots, X^T$ of a continuous random variable $X$. We denote $X^{(1:T)} \leqslant X^{(2:T)} \leqslant \cdots < X^{(T:T)}$ the corresponding order statistics and define the associated $m$-spacings ($1\leqslant m < T$) as the sequence of non-negative differences $X^{(i+m:T)}-X^{(i:T)}$, for $1 \leqslant i \leqslant T-m$. 

In a first step, we build a 1-spacing estimator of $H_\alpha^{\exp}(X)$ because the case $m=1$ has a natural relation to a sample-spacings estimator of the density $f_X$. 

First, recall that the order statistics $Y^{(1:T)},\ldots,Y^{(T:T)}$ of a uniform $\mathcal{U}(0,1)$ random variable $Y$ follow a Beta distribution (\hyperref[Arnold]{Arnold et al. 1992}). In particular, 
\begin{equation*}
\mathbb{E}\big(Y^{(i:T)}\big)  = \frac{i}{T+1}\;.
\end{equation*}

Let us now map $X^1,\ldots,X^T$ through $F_X$ to obtain $T$ $\mathcal{U}(0,1)$ i.i.d. random variables $Y^i:=F_X(X^i)$. Obviously, the sequence $F_X(X^{(1:T)}),\ldots,F_X(X^{(T:T)})$ agrees with the order statistics $Y^{(1:T)},\ldots, Y^{(T:T)}$, leading to:
\begin{equation*}
\mathbb{E}\big(Y^{(i:T)}\big)  =\mathbb{E}\big(F_X\big(X^{(i:T)}\big)\big) = \mathbb{P}\big(X \leqslant X^{(i:T)}\big) = \frac{i}{T+1}\;.
\end{equation*}
Hence, the expected probability mass between two order statistics $X^{(i:T)}\leqslant X^{(i+1:T)}$ is
\begin{equation}\label{ExpArea}
\mathbb{E}\left(\int_{X^{(i:T)}}^{X^{(i+1:T)}} f_X(x) dx\right)= \frac{1}{T+1}\;.
\end{equation}

One can use this key observation to obtain an estimator $\widehat{f}_X$ of $f_X$ being told $T$ order statistics. Indeed, one can thus approximate $f_X(x)$ between two successive order statistics $X^{(i:T)},X^{(i+1:T)}$ by a constant $k_i$ such that the corresponding probability mass
\begin{equation*}
\begin{aligned}
\int_{X^{(i:T)}}^{X^{(i+1:T)}} \widehat{f}_X(x)dx &= \int_{X^{(i:T)}}^{X^{(i+1:T)}} k_i dx 
\\ &= k_i\big(X^{(i+1:T)}-X^{(i:T)}\big)
\end{aligned}
\end{equation*}
agrees with the expected probability mass in \eqref{ExpArea}. Denoting $X^{(0:T)} \coloneqq \inf X$ and $X^{(T+1:T)} := \sup X$, this yields
\begin{equation*}
k_i = \frac{1}{(T+1)(X^{(i+1:T)}-X^{(i:T)})}\;
\end{equation*}
for $X^{(i:T)} < x \leqslant X^{(i+1:T)}$. As the $T+1$ spacings form a partition of
$\big[X^{(0:T)},X^{(T+1:T)}\big]$, one can approximate the density $f_X$ by 
\begin{equation}\label{DensityEstimate}
\widehat{f}_X(x) =\sum_{i=0}^T \ind_{\{X^{(i:T)}<x\leqslant X^{(i+1:T)}\}}k_i\;.
\end{equation}
This estimator corresponds to the histogram composed of $T+1$ bins with bounds $[X^{(i:T)},X^{(i+1:T)}]$, $0 \leqslant i\leqslant T$, and with height such that the area of each bin is equal to $1/(T+1)$.
 
From this density estimator, one can derive a 1-spacing plug-in estimator of $H_\alpha^{\exp}$ as follows.
\begin{prop}\label{1spacingExpRenyi}
Let the density $f_X$ of a continuous random variable $X$ be approximated by \eqref{DensityEstimate}, then the 1-spacing plug-in estimator of $H_\alpha^{\exp}(X)$ is given by
\begin{equation}\label{OneSpacing}
\Bigg(\frac{1}{T+1} \sum_{i=0}^{T} \Big((T+1)\big(X^{(i+1:T)}-X^{(i:T)}\big)\Big)^{1-\alpha}\Bigg)^{\frac{1}{1-\alpha}}\;.
\end{equation}
\end{prop}
\begin{proof}
The 1-spacing estimator of $\int (f_X(x))^\alpha dx$ becomes
\begin{align*}  
&\int (\widehat{f}_X(x))^\alpha dx \\
&= \sum_{i=0}^T \int_{X^{(i:T)}}^{X^{(i+1:T)}} (\widehat{f}_X(x))^\alpha dx \\
&= \sum_{i=0}^T \frac{1}{\big((T+1)(X^{(i+1)}-X^{(i)})\big)^\alpha}  \int_{X^{(i:T)}}^{X^{(i+1:T)}} dx  \\
&= \frac{1}{T+1} \sum_{i=0}^T \Big((T+1)\big(X^{(i+1:T)}-X^{(i:T)}\big)\Big)^{1-\alpha}\;,
\end{align*}
resulting in \eqref{OneSpacing}.
\end{proof}
The estimator in \eqref{OneSpacing} can not be used as such because, in general, we do not know $X^{(0:T)}$ and $X^{(T+1:T)}$, i.e. the true support of $X$. Following \hyperref[Miller]{Learned-Miller and Fisher (1993)}, we therefore disregard the values below $X^{(1:T)}$ and above $X^{(T:T)}$, and compensate this by a factor $\frac{T+1}{T-1}$, yielding the final approximation
\begin{equation*}
\begin{aligned}
&\widehat{H}_\alpha^{\exp}(1,T) := \\ &\Bigg(\frac{1}{T-1} \sum_{i=1}^{T-1}\Big((T+1)\big(X^{(i+1:T)}-X^{(i:T)}\big)\Big)^{1-\alpha}\Bigg)^{\frac{1}{1-\alpha}}\;.
\end{aligned}
\end{equation*}

As detailed by \hyperref[Miller]{Learned-Miller and Fisher (1993)} in the specific case of Shannon entropy, the 1-spacing estimator suffers from high variance. To reduce the asymptotic variance, one can consider a $m$-spacings estimator where the $m$-spacings overlap. The counterpart of $k_i$ becomes
\begin{equation*}
k_i(m) := \frac{m}{(T+1)(X^{(i+m:T)}-X^{(i:T)})}
\end{equation*}
for $X^{(i:T)} < x \leqslant X^{(i+m:T)}$. However, because the $m$-spacings overlap, they do not form a partition of $\big[X^{(0:T)},X^{(T+1:T)}\big]$ anymore (the same $x$ can fall in more than one $m$-spacing), hence we lose the correspondence with the density estimator as a weighted sum of indicators in \eqref{DensityEstimate}. Still, from the definition of $k_i(m)$, we can consider this extension of $\widehat{H}_\alpha^{\exp}(1,T)$:
\begin{equation*}
\begin{aligned}
&\widehat{H}_\alpha^{\exp}(m,T) := \\ &\Bigg(\frac{1}{T-m}  \sum_{i=1}^{T-m}  \bigg(\frac{T+1}{m}\big(X^{(i+m:T)}-X^{(i:T)}\big)\bigg)^{1-\alpha}\Bigg)^{\frac{1}{1-\alpha}}\;,
\end{aligned}
\end{equation*}
which corresponds to \eqref{mspacingoriginal}. Taking the limit $\alpha \rightarrow 1$, we recover the estimator in \eqref{ExpShannonEst}: 
\begin{equation*}
\begin{aligned}
&\widehat{H}_1^{\exp}(m,T) := \\ &\exp\Bigg(\frac{1}{T-m} \sum_{i=1}^{T-m}\ln\bigg(\frac{T+1}{m}\big(X^{(i+m:T)}-X^{(i:T)}\big)\bigg)\Bigg)\;.
\end{aligned}
\end{equation*}

\section*{References}

\setlength{\parindent}{0pt}

\hangindent=1.5em
\hangafter=1
\phantomsection
\label{Abbas}
Abbas, A. (2006). Maximum entropy utility. \textit{Operations Research}, 54(2), 277--290.


\hangindent=1.5em
\hangafter=1
\phantomsection
\label{Adcock}
Adcock, C. (2014). Mean–variance–skewness efficient surfaces, Stein’s lemma and the multivariate extended skew-student distribution. \textit{European Journal of Operational Research}, 234(2), 392--401.

\hangindent=1.5em
\hangafter=1 \phantomsection
\label{Ardia}
Ardia, D., Bolliger, G., Boudt, K., \& Gagnon-Fleury, J. (2017). The impact of covariance misspecification in risk-based portfolios. \textit{Annals of Operations Research}, 254(1-2), 1--16.

\hangindent=1.5em
\hangafter=1 \phantomsection
\label{Arnold}
Arnold, B., Balakrishnan, N., \& Nagaraja, H. (1992). \textit{A first course in order statistics}. New York: John Wiley \& Sons.

\hangindent=1.5em
\hangafter=1 \phantomsection
\label{Artzner}
Artzner, P., Delbaen, F., Eber, J., \& Heath, D. (1999). Coherent measures of risk. \textit{Mathematical Finance}, 9(3), 203--228.

\hangindent=1.5em
\hangafter=1 \phantomsection
\label{Behr}
Behr, P., Guettler, A., \& Miebs, F. (2013). On portfolio optimization: Imposing the right constraints. \textit{Journal of Banking and Finance}, 37, 1232--1242.

\hangindent=1.5em
\hangafter=1 \phantomsection
\label{Beirlant}
Beirlant, J., Dudewicz, E., Gyofi, L., \&  van der Meulen, E. (1997). Non-parametric entropy estimation: An overview. \textit{International Journal of Mathematical and Statistical Sciences}, 6(1), 17--39.

\hangindent=1.5em
\hangafter=1 \phantomsection
\label{Bera}
Bera, A., \& Park, S. (2008). Optimal portfolio diversification using the maximum entropy principle. \textit{Econometric Reviews}, 27(4–6), 484--512.


\hangindent=1.5em
\hangafter=1 \phantomsection
\label{BoudtPeterson} 
Boudt, K., Peterson, B., \& Croux, C. (2008). Estimation and decomposition of downside risk for portfolios with non-normal returns. \textit{Journal of Risk}, 11, 79--103.

\hangindent=1.5em
\hangafter=1 \phantomsection
\label{Campbell} 
Campbell, L. (1966). Exponential entropy as a measure of extent of a distribution. \textit{Z. Wahrsch.},
5, 217--225.

\hangindent=1.5em
\hangafter=1 \phantomsection
\label{Carroll}
Carroll, R., Conlon, T., Cotter, J., \& Salvador, E. (2017). Asset allocation with correlation: A composite trade-off. \textit{European Journal of Operational Research}, 262(3), 1164--1180.


\hangindent=1.5em
\hangafter=1 \phantomsection
\label{Chen}
Chen, L., He, S., \& Zhang, S. (2011). When all risk-adjusted performance measures are
the same: In praise of the Sharpe ratio. \textit{Quantitative Finance}, 11(10), 1439--1447.


\hangindent=1.5em
\hangafter=1 \phantomsection
\label{Cont}
Cont, R. (2001). Empirical properties of asset returns: Stylized facts and statistical issues. \textit{Quantitative
Finance}, 1, 223--236.

\hangindent=1.5em
\hangafter=1 \phantomsection
\label{Cover} 
Cover, T., \& Thomas, J. (2006). \textit{Elements of information theory (2nd ed.)}. New Jersey: Wiley.

\hangindent=1.5em
\hangafter=1 \phantomsection
\label{Danielsson}
Daníelsson, J., Jorgensenb, B., Samorodnitskyc, G., Sarma, M., \& de Vries, C. (2013). Fat tails, VaR and subadditivity. \textit{Journal of Econometrics}, 172, 283--291.

\hangindent=1.5em
\hangafter=1 \phantomsection
\label{DeMiguel}
DeMiguel, V., Garlappi, L., \& Uppal, R. (2009a). Optimal versus naive diversification: How inefficient is the 1/N portfolio strategy? \textit{The Review of Financial Studies}, 22(5), 1915--1953.

\hangindent=1.5em
\hangafter=1 \phantomsection
\label{DeMiguel1}
DeMiguel, V., Garlappi, L., Nogales, F., \& Uppal, R. (2009b). A generalized approach to portfolio optimization: Improving performance by constraining portfolio norms. \textit{Management Science}, 55(5), 798--812.

\hangindent=1.5em
\hangafter=1 \phantomsection
\label{DeMiguelNogales}
DeMiguel, V., \& Nogales, F. (2009). Portfolio selection with robust estimation. \textit{Operations Research}, 57, 560--577.

\hangindent=1.5em
\hangafter=1 \phantomsection
\label{Dionisio}
Dionisio, A., Menezes, R., \& Mendes, A. (2006). An econophysics approach to analyse uncertainty in financial markets: An application to the Portuguese stock market. \textit{The European Physical Journal B}, 50, 161--164.



\hangindent=1.5em
\hangafter=1 \phantomsection
\label{VanEs} 
van Es, B. (1992). Estimating functionals related to a density by a class of statistics based on spacings. \textit{Scandinavian Journal of Statistics}, 19(1), 61--72.

\hangindent=1.5em
\hangafter=1 \phantomsection
\label{Fabozzi}
Fabozzi, F., Huang, D., \& Zhou, G. (2010). Robust portfolios: Contributions from operations research and finance. \textit{Annals of Operations Research}, 176(1), 191--220.

\hangindent=1.5em
\hangafter=1 \phantomsection
\label{Fang}
Fang, K., \& Zhang, Y. (1990). \textit{Generalized multivariate analysis}. New York: Springer.

\hangindent=1.5em
\hangafter=1 \phantomsection
\label{Flores}
Flores, Y., Bianchi, R., Drew, M., \& Tr\"uck, S. (2017). The diversification delta: A different perspective. \textit{Journal of Portfolio Management}, 43(4), 112--124.

\hangindent=1.5em
\hangafter=1 \phantomsection
\label{Hampel} 
Hampel, F., Ronchetti, E., Rousseeuw, P., \& Stahel, W. (1986). \textit{Robust statistics: The approach based on influence functions}. New York: Wiley.

\hangindent=1.5em
\hangafter=1 \phantomsection
\label{Hegde} 
Hegde, A., Lan, T., \& Erdogmus, D. (2005). Order statistics based estimator for Rényi entropy. \textit{IEEE Workshop on Machine Learning for Signal Processing},
335--339.


\hangindent=1.5em
\hangafter=1
\phantomsection
\label{Hyvarinen}
Hyv\"arinen, A., Karhunen, J., \& Oja, E. (2001). \textit{Independent component analysis}. New York: John Wiley \& Sons.


\hangindent=1.5em
\hangafter=1 \phantomsection
\label{Johnson} 
Johnson, O., \& Vignat, C. (2007). Some results concerning maximum Rényi entropy distributions. \textit{Annales de l'Institut Henri-Poincaré (B) Probab. Statist.}, 43(3), 339--351.

\hangindent=1.5em
\hangafter=1 \phantomsection
\label{Jorion}
Jorion, P. (1986). Bayes-Stein estimation for portfolio analysis. \textit{Journal of Financial and Quantitative Analysis}, 21(3), 279--292.

\hangindent=1.5em
\hangafter=1 \phantomsection
\label{Jose}
Jose, V., Nau, R., \& Winkler, R. (2008). Scoring rules, generalized entropy, and utility maximization. \textit{Operations Research}, 56(5), 1146--1157. 

\hangindent=1.5em
\hangafter=1 \phantomsection
\label{Jurczenko}
Jurczenko, E., \& Maillet, B. (2006). \textit{Multi-moment asset allocation
and pricing models}. West Sussex: John Wiley \& Sons.

\hangindent=1.5em
\hangafter=1 \phantomsection
\label{Kolm}
Kolm, P., T\"ut\"unc\"u, R., \& Fabozzi, F. (2014). 60 years of portfolio optimization: Practical challenges and current trends. \textit{European Journal of Operational Research}, 234(2), 356--371.

\hangindent=1.5em
\hangafter=1 \phantomsection
\label{Koski}
Koski, T., \& Persson, L. (1992). Some properties of generalized exponential entropies with application to data compression. \textit{Information Sciences}, 62, 103--132.

\hangindent=1.5em
\hangafter=1 \phantomsection
\label{Miller}
Learned-Miller, E., \& Fisher, J. (2003). ICA using spacings estimates of entropy. \textit{Journal of Machine Learning Research}, 4, 1271--1295.

\hangindent=1.5em
\hangafter=1 \phantomsection
\label{Ledoit1} 
Ledoit, O., \& Wolf, M. (2003). Improved estimation of the covariance matrix of stock returns with an application to portfolio selection. \textit{Journal of Empirical Finance}, 10, 603--621.

\hangindent=1.5em
\hangafter=1 \phantomsection
\label{Ledoit2} 
Ledoit, O., \& Wolf, M. (2004a). A well-conditioned estimator for large-dimensional covariance matrices. \textit{Journal of Multivariate Analysis}, 88, 365--411.

\hangindent=1.5em
\hangafter=1 \phantomsection
\label{Ledoit3} 
Ledoit, O., \& Wolf, M. (2004b). Honey, I shrunk the sample covariance matrix. \textit{Journal of Portfolio Management}, 30, 110--119.


\hangindent=1.5em
\hangafter=1 \phantomsection
\label{Levy}
Levy, H., \& Levy, M. (2014). The benefits of differential variance-based constraints in portfolio optimization. \textit{European Journal of Operational Research}, 234, 372--381.

\hangindent=1.5em
\hangafter=1 \phantomsection
\label{Markowitz}
Markowitz, H. (1952). Portfolio selection. \textit{Journal of Finance}, 7(1), 77--91.

\hangindent=1.5em
\hangafter=1
\phantomsection
\label{Martellini}
Martellini, L., \& Ziemann, V. (2010). Improved estimates of higher-order comoments and implications for portfolio selection. \textit{Review of Financial Studies}, 23(4), 1467--1502.


\hangindent=1.5em
\hangafter=1
\phantomsection
\label{Ormos}
Ormos, M., \& Zibriczky D. (2014). Entropy-based financial asset pricing. \textit{Plos One}, 9(12), e115742.

\hangindent=1.5em
\hangafter=1 \phantomsection
\label{Pezier}
Pézier, J. (2004). Risk and risk aversion. In Alexander, C., \& Sheedy, E. \textit{The professional risk managers' handbook}. Wilmington DE: PRMIA Publications.

\hangindent=1.5em
\hangafter=1 \phantomsection
\label{Pham}
Pham, D., \& Vrins, F. (2005) Local minima of information-theoretic criteria in blind source separation. \textit{IEEE Signal Processing Letters}, 12(11), 788--791.

\hangindent=1.5em
\hangafter=1 \phantomsection
\label{Pham1} 
Pham, D., Vrins, F., \& Verleysen, M. (2008). On the risk of using Rényi’s entropy for blind source separation. \textit{IEEE Transactions on Signal Processing}, 56(10), 4611--4620.

\hangindent=1.5em
\hangafter=1 \phantomsection
\label{Philippatos} 
Philippatos, G., \& Wilson, C. (1972).
Entropy, market risk, and the selection of efficient portfolios. \textit{Applied Economics}, 4(3), 209--220.

\hangindent=1.5em
\hangafter=1 \phantomsection
\label{Qi}
Qi, Y., Steuer, R., \& Wimmer, M. (2017). An analytical derivation of the efficient surface in portfolio selection with three criteria. \textit{Annals of Operations Research}, 251(1-2), 161--177.

\hangindent=1.5em
\hangafter=1 \phantomsection
\label{Renyi}
Rényi, A. (1961). On measures of entropy and information. \textit{Fourth Berkeley Symposium on Mathematical Statistics and Probability}, 547--561.

\hangindent=1.5em
\hangafter=1 \phantomsection
\label{Rockafellar2}
Rockafellar, R., Uryasev, S., \& Zabarankin, M. (2006). Generalized deviations in risk analysis.
\textit{Finance and Stochastics}, 10, 51--74.

\hangindent=1.5em
\hangafter=1 \phantomsection
\label{Royston}
Royston, J. (1982). Expected normal order statistics (exact and approximate). \textit{Journal of the Royal Statistical Society Series C (Applied Statistics)}, 31(2), 161--165.

\hangindent=1.5em
\hangafter=1 \phantomsection
\label{Sbuelz}
Sbuelz, A., \& Trojani, F. (2008). Asset prices with locally constrained-entropy recursive multiple-priors utility. \textit{Journal of Economic Dynamics and Control}, 32(11), 3695--3717.


\hangindent=1.5em
\hangafter=1 \phantomsection
\label{Scutella}
Scutellà, M., \& Recchia, R. (2013). Robust portfolio asset allocation and risk measures. \textit{Annals of Operations Research}, 204(1), 145--169.

\hangindent=1.5em
\hangafter=1 \phantomsection
\label{Shannon}
Shannon, C. (1948). A mathematical theory of communication. \textit{Bell Systems Technical Journal}, 27, 379--423 and 623--656.

\hangindent=1.5em
\hangafter=1\phantomsection
\label{Ugray}
Ugray, Z., Lasdon, L., Plummer, J. Glover, F., Kelly, J., \& Martí, R. (2007). Scatter search and local NLP solvers: A multistart framework for global optimization. \textit{INFORMS Journal on Computing}, 19(3), 328--340.

\hangindent=1.5em
\hangafter=1 \phantomsection
\label{Vanduffel}
Vanduffel, S., \& Yao, J. (2017). A stein type lemma for the multivariate generalized hyperbolic
distribution. \textit{European Journal of Operational Research}, 261(2), 606--612.

\hangindent=1.5em
\hangafter=1 \phantomsection
\label{Vasicek}  
Vasicek, O. (1976). A test for normality based on entropy. \textit{Journal of the Royal Statistical Society  Series B (Methodological)}, 38(1), 54--59.

\hangindent=1.5em
\hangafter=1 \phantomsection
\label{Vermorken}
Vermorken, M., Medda, F., \& Schroder, T. (2012). The diversification delta: A higher-moment measure for portfolio diversification. \textit{Journal of Portfolio Management}, 39(1), 67--74.

\hangindent=1.5em
\hangafter=1 \phantomsection
\label{Vrins}
Vrins, F., Pham, D., \& Verleysen, M. (2007). Mixing and non-mixing local minima of the entropy contrast for blind source separation. \textit{IEEE Transactions on Information Theory}, 53(3), 1030--1042.

\hangindent=1.5em
\hangafter=1 \phantomsection
\label{Wachowiak}
Wachowiak, M., Smolikova, R., Tourassi, G., \& Elmaghraby, A. (2005). Estimation of generalized entropies with sample spacing. \textit{Pattern Analysis and Applications}, 8, 95--101.

\hangindent=1.5em
\hangafter=1 \phantomsection
\label{Yang}
Yang, J., \& Qiu, W. (2005). A measure of risk and a decision-making model based on expected utility and entropy. \textit{European Journal of Operational Research}, 164(3), 792--799.


\hangindent=1.5em
\hangafter=1 \phantomsection
\label{Zhou}
Zhou, R., Cai, R., \& Tong, G. (2013). Applications of entropy in finance: A review. \textit{Entropy}, 15, 4909--4931.

\hangindent=1.5em
\hangafter=1 \phantomsection
\label{Zografos}
Zografos, K., \& Nadarajah, S. (2003). Formulas for Rényi information and related measures for univariate distributions. \textit{Information Sciences}, 155(1-2), 119--138.

\hangindent=1.5em
\hangafter=1 \phantomsection
\label{Zografos1}
Zografos, K., \& Nadarajah, S. (2005). Expressions for Rényi and Shannon entropies for multivariate distributions. \textit{Statistics \& Probability Letters}, 71(1), 71--84.

\end{document}